\documentclass[]{llncs}

\usepackage[pdftex]{graphicx}
\usepackage{amssymb}
\usepackage{xspace}
\usepackage{array}
\usepackage{lineno}

\usepackage{color}

\newcommand{\akifixvset}{Q}
\newcommand{\akianspset}{K}
\usepackage{amsmath}
\usepackage{algorithm}
\usepackage{algorithmic}
\DeclareMathOperator*{\argmax}{arg\,max}
\newcommand{\sMax}[1]{\eta(\vec{s}, #1)}
\newcommand{\TRUE}{{\sf true}}
\newcommand{\FALSE}{{\sf false}}
\newcommand{\CD}{{\sc Competitive Diffusion}}

\def\Gada{a'_1}
\def\Gadb{a'_2}
\def\Gadd{a'_n}

\def\Gdda{a''_1}
\def\Gddb{a''_2}
\def\Gddd{a''_n}

\def\bcenter{b}
\def\bw{\lambda}

\def\ba{b_j}
\def\bb{b_{j, 1}}
\def\bc{b_{j, 2}}
\def\bd{b_{j, 3}}

\def\ps{\alpha}

\def\umin{\mu}
\def\suki{\nu}

\begin{document}

\title{Computational Complexity of Competitive Diffusion on (Un)weighted Graphs}

\author{%
Takehiro Ito\inst{1} \and
Yota Otachi\inst{2} \and
Toshiki Saitoh\inst{3} \and
Hisayuki Satoh\inst{1} \and
Akira Suzuki\inst{1} \and
Kei Uchizawa\inst{4} \and
Ryuhei Uehara\inst{2} \and
Katsuhisa Yamanaka\inst{5} \and
Xiao Zhou\inst{1}
}

\institute{
Tohoku University,\\
\email{\{takehiro, h.satoh, a.suzuki, zhou\}@ecei.tohoku.ac.jp}
\and
Japan Advanced Institute of Science and Technology,\\
\email{\{otachi, uehara\}@jaist.ac.jp}
\and
Kobe University,\\
\email{saitoh@eedept.kobe-u.ac.jp}
\and
Yamagata University,\\
\email{uchizawa@yz.yamagata-u.ac.jp}
\and
Iwate University,\\
\email{yamanaka@cis.iwate-u.ac.jp}
}

\maketitle
\begin{abstract}
Consider an undirected graph modeling
a social network, where
the vertices represent users, and the edges do connections among them.
In the competitive diffusion game, each of a number of players chooses
a vertex as a seed to propagate his/her opinion, and then it spreads along
the edges in the graphs. The objective of every player is to maximize
the number of vertices the opinion infects.
In this paper, we investigate a computational problem of asking whether a pure Nash equilibrium 
exists in the competitive diffusion game on unweighed and weighted graphs, and
present several negative and positive results.
We first prove that the problem is W[1]-hard when parameterized by the
number of players even for unweighted graphs. We also show that
the problem is NP-hard even for series-parallel graphs with positive integer weights,
and  is NP-hard even for forests with arbitrary integer weights.
Furthermore, we show that the problem for forest of paths with arbitrary weights
is solvable in pseudo-polynomial time;
and it is solvable in quadratic time if a given graph is unweighted. We also
prove that the problem for chain, cochain, and threshold graphs with
arbitrary integer weights
is solvable in polynomial time.
 \end{abstract}

\section{Introduction}
Social networks plays an important role in
human interactions in recent years, through which we can propagate
ideas, findings, innovations and trends.
An action some influential users exhibit spreads over social networks, and
greatly affects public opinion formation.
The dynamics of such diffusion process
is a major research subject on game-theoretic analysis of social networks.
In particular,
Alon \emph{et al.}~\cite{Alon2010221} focused on a setting of a diffusion process
where a number of players compete to propagate their opinions on a topic
over a social network. In this setting,
each player chooses a user who acts as a seed of the opinion,
and then the opinion  spreads along
the connections in a social network in discrete time step.
The objective of every player
is to maximize the number of users their opinions infect.
This setting models, for example, a situation where a number of firms
compete in a specific market, and want to advertise their products via
viral marketing; each firm chooses an influential
user to maximize the number of users who buy their product.
Alon \emph{et al.} formalized this process as the \emph{competitive diffusion game}
on unweighted graphs. In this paper, we generalize the game to weighted graphs,
and investigate its computational complexity.
We below define the game more formally.

\begin{figure}[t]
\begin{center}
\includegraphics[width=12cm]{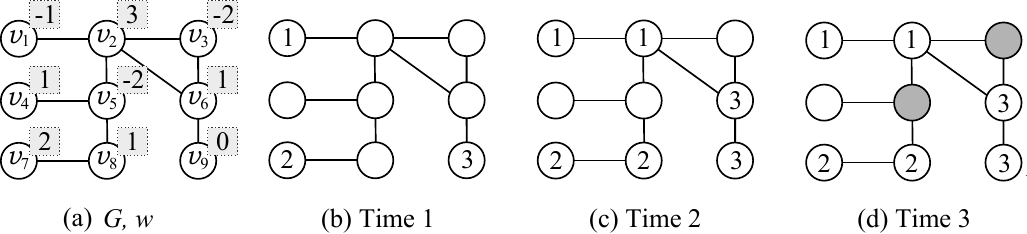}
\end{center}
\caption{Example of competitive diffusion with $k=3$ players. (a) The graph $G$ and weight $w$; numbers in the gray squares are weights. (b) $p_1, p_2$ and $p_3$ choose $v_1, v_7$ and $v_9$ in $G$, respectively;
thus a strategy profile $\vec{s}=(v_1, v_7, v_9)$.
(c) Each player dominates the neighbor. (d) The game ends; the two gray vertices are neutral. Consequently, $U_1(\vec{s})=2, U_2(\vec{s})=3$ and $U_3(\vec{s})=1$.}
\label{Fig:game}
\end{figure}
In the competitive diffusion game, a social network is modeled by an undirected graph $G =(V, E)$, where the vertex set $V$ represents users, the edge set
$E$ does their connections; thus, their neighbors do friends they communicate.
After some user exhibits an opinion, it spreads along the
edges in the graph.
We then define a weight function $w: V \to \mathbb{Z}$ which represents
a level of importance of each user; we admit that a weight on a vertex can
be negative to express, for example, a very demanding customer.
For a positive integer $k$, we define $[k] = \{ 1, 2, \dots , k\}$, and
call the $k$ players  $p_1, p_2, \dots , p_k$.
The competitive diffusion game $(k, G, w)$ proceeds as follows.
(See Fig.~\ref{Fig:game}(a)-(d) for an explicit example.)
At time one, each player chooses a vertex in $V$;
suppose a player $p_i$, $i \in [k]$, chooses a vertex $v\in V$.
If any other player $p_j$, $i \ne j$,
does not choose the vertex $v$, then $p_i$ \emph{dominates} $v$; and otherwise (that is, if there exists a player $p_j$, $i \ne j$, who chooses $v$), $v$ becomes a \emph{neutral} vertex which
no player can dominate in the following time steps.
For each time $t$, $t \ge 2$, a vertex $v \in V$ is dominated by a player $p_i$ at time $t$
if (i) $v$ is neither neutral nor dominated by any player by time $(t-1)$,
and (ii) $v$ has a neighbor dominated
by $p_i$, but does not have a neighbor dominated by any player $p_j$,
$i \ne j$. If $v$ satisfies (i) and there are
two or more players who dominate neighbors of $v$, then $v$ becomes
a neutral vertex at time $t$. The game ends when no player can dominate
a vertex any more. 

Let $\vec{s}=(s^{(1)}, s^{(2)}, \dots , s^{(k)})\in V^k$ be the vector of vertices which the players
choose at the beginning of the game. We call $\vec{s}$ a \emph{strategy profile}.
For every $i\in [k]$,
we define a \emph{utility} $U_i(\vec{s})$ of $p_i$ for $\vec{s}$ as the sum of the weights
of the vertices which $p_i$ dominates at the end. (See Fig.~\ref{Fig:game}(d).)

We note that a neutral vertex is a key notion of the
competitive diffusion game; it sometimes gives critical effect on the result.
(See Fig~\ref{Fig:neutral}.)
This contrasts to a similar game, called Voronoi game, where
a player can dominate all the nearest vertices; if there is a vertex whose
distances to strategies of two or more players tie,
then they do not dominate but share
the vertex~\cite{Durr_Thang_2007,Feldmann:2009,TeramotoDU11}.
\begin{figure}[t]
\begin{center}
\includegraphics[width=6cm]{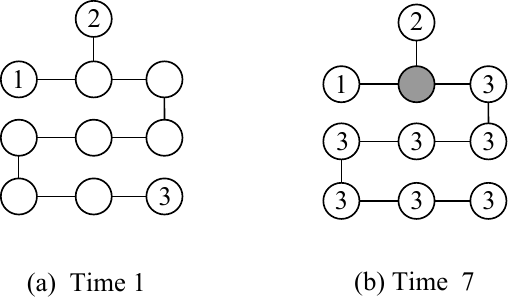}
\end{center}
\caption{The vertex $v_3$ becomes neutral at time 2, and consequently, $p_3$
dominates $v_4$ at time 7.}
\label{Fig:neutral}
\end{figure}

For an index $i \in [k]$,
we define $(\vec{s}_{-i}, v')$ as a strategy profile such that
$p_i$ chooses $v'$ instead of $s^{(i)}$, but any other player $p_j$, $i \ne j$, chooses
$s^{(j)}$:
$(\vec{s}_{-i}, v) = (s^{(1)}, s^{(2)}, \dots , s^{(i-1)}, v', s^{(i+1)},\dots , s^{(k)})$.
For simplicity, we write $U_i(\vec{s}_{-i}, v')$ for $U_i((\vec{s}_{-i}, v'))$.
Then, if $\vec{s}$ satisfies
\begin{eqnarray}\label{eq:condition_on_Nash}
U_i(\vec{s}_{-i}, v') \le U_i(\vec{s})
\end{eqnarray}
for every $i\in [k]$ and every $v'\in V$, we say that $\vec{s}$ is a
\emph{(pure) Nash equilibrium}. The strategy profile given in Fig.~\ref{Fig:game}(b)
is, in fact, a Nash equilibrium. 

Since Alon \emph{et al.} proposed the competitive diffusion game
in \cite{Alon2010221},
several results are obtained for the game on unweighted graphs.
In the same paper,
Alon \emph{et al.} observed that there exists an unweighted graph $G$ of
diameter three such that $(2, G, w)$ does not admit
a Nash equilibrium~\cite{Alon2010221}.
Takehara improved their result and showed that there exists
an unweighted graph $G$ of
diameter two such that $(2, G, w)$ does not admit a
Nash equilibrium~\cite{Takehara201259}.
(Besides, they refute a theorem given in \cite{Alon2010221}.)
Small and Mason considered the case where a social network has a tree structure,
and show that, for any unweighted tree $G$,
$(2, G, w)$ has a Nash equilibrium~\cite{SM13}.

In this paper, we consider a problem, named {\CD}, of deciding whether,
given the number $k$, a graph $G$ and weight function $w$, the
competitive diffusion game $(k, G, w)$ has a Nash equilibrium.
Recently, Etesami and Baser studied the problem,
and showed that {\CD} is a NP-complete problem~\cite{RB14}.
We strengthen their result, and generalize it in terms of weight functions.
More formally, we obtain the following
three negative results:
(i) {\CD} is W[1]-hard when parameterized by the number of players
even for unweighted graphs;
(ii) {\CD} is NP-complete even for series-parallel graphs with positive integer weights;
and (iii) {\CD} is NP-complete even for forests with arbitrary integer weights.
Our results imply that the computational complexity of {\CD} may be
much influenced by weights on vertices.
We note that the results given in~\cite{RB14} does not imply ours.

Furthermore, we obtain several positive results on {\CD}.
For forests of paths, we prove that {\CD}
is solvable in pseudo-polynomial time.
In particular, we give a quadratic-time algorithm for forests of
unweighted paths.
For chain, cochain, and threshold graphs with arbitrary integer weights,
we also provide
polynomial time algorithms.
To the best of our knowledge,
no nontrivial algorithm is known even for unweighted version of
{\CD}.
 
The rest of the paper is organized as follows.
In~ Section~\ref{sec:Complexity}, we present our hardness
results for {\CD}.
In Section~\ref{sec:AlgForPaths}, we give algorithms for forests of paths.
In Section~\ref{sec:AlgForCochain}, we provide an algorithm for
chain, cochain, and threshold graphs.

\section{Complexity of {\CD}}\label{sec:Complexity}

In this section, we show that {\CD} is basically
hard computational problem.
In Section~\ref{ssec:W[1]_UW},
we show W[1]-hardness for unweighted graphs.
In Sections~\ref{ssec:Nonnegative}
and~\ref{ssec:NP_Arbitrary},
we show NP-completeness for series-parallel graphs with nonnegative integer weights
and NP-completeness for forests with arbitrary integer weights, respectively.

\subsection{Unweighted Graphs}\label{ssec:W[1]_UW}

In this section, we prove that
the problem is $W[1]$-hard for the number of players,
as in the following theorem.

\begin{theorem}
{\CD} is $W[1]$-hard even for unweighted graphs
when parameterized by the number of players.
\end{theorem}

We prove the theorem by reduction from a well-known $W[1]$-hard
problem, {\sc Independent Set}~\cite{Flum06}.
Given a graph $G = (V, E)$ and a positive integer $k$,
{\sc Independent Set} asks whether there exists an independent set $I$
of size at least $k$, where a set $I$($\subseteq V$) is called an independent
set if there is no pair of vertices $u, v \in I$ such that $(u, v) \in E$.
Below we show that, given $G$ and $k$, one can construct a competitive diffusion game $(k+3, G', w')$ such that $G$ has an independent set $I$ of
size $|I| \ge k$
if and only if  $(k+3, G', w')$ has a pure Nash equilibrium, where
$G'$ and $w'$ is a graph and weight obtained from $G$.
Let $n =|V|$ and $m=|E|$, and let  $\delta_v$ be the degree of $v$
for every $v\in V$.

\begin{figure}[t]
\begin{center}
\includegraphics[width=85mm,clip]{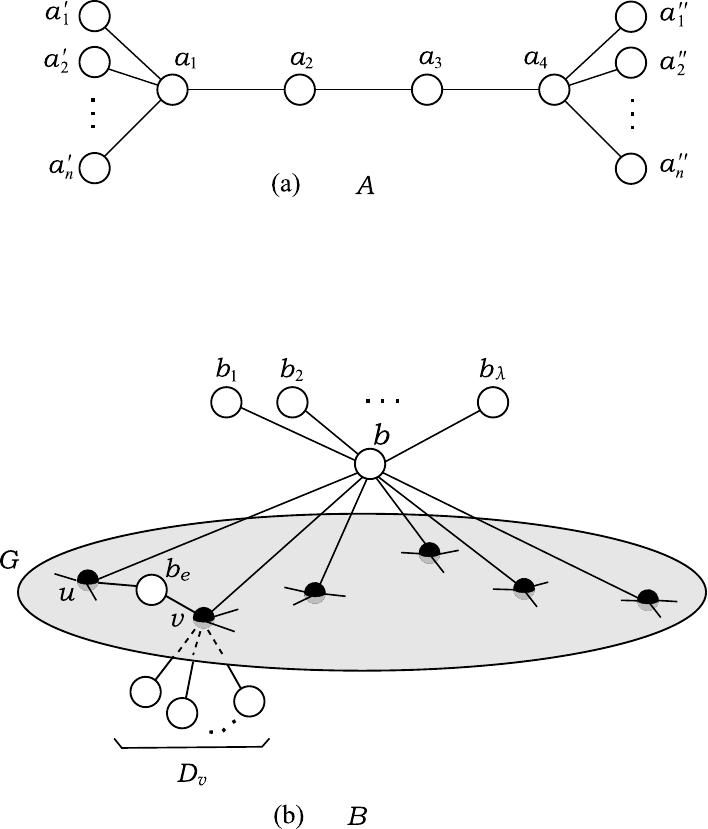}
\caption{Unweighted graph $G'$. (a) Graph $A$. (b) Graph $B$; The black vertices originate from $G$, and compose $V$.}
\label{fig:general}
\end{center}
\end{figure}
\smallskip\noindent
[Construction of $G'$ and $w'$]

The graph $G'$ consists of two disjoint component graphs $A=(V_A, E_A)$
and $B=(V_B, E_B)$.

We obtain the graph $A$ as follows.
We construct a path of four vertices $a_1, a_2, a_3, a_4$; and
make $2n$ vertices $\Gada, \Gadb, \dots , \Gadd$ and
$\Gdda,\Gddb, \dots ,\Gddd$. Then
we connect $a_1$ to $\Gada, \Gadb, \dots , \Gadd$, and
connect $a_4$ to $\Gdda,\Gddb, \dots ,\Gddd$.
See Fig.~\ref{fig:general}(a).
We obtain the graph $B$ from the original graph $G$ as follows.
For every edge $e=(u, v)\in E$, we add a vertex $b_e$
subdividing $e$. Then, for each $v \in V$, we make
a set $D_v$ of $n-\delta_v$ vertices, and connect
$v$ to every $u \in D_v$.
Lastly we make a vertex $b$ and $\bw$ vertices $b_1, b_2, \dots , b_{\bw}$,
where $\bw$ is a sufficiently large number satisfying $\bw = \Theta (n^3)$,
and connect $b$ to every $v \in V$, and connect $b$ to
$b_1, b_2, \dots , b_{\bw}$. 

Thus, $G'=(V', E')$ is such a graph that $V' = V_A \cup V_B$ and
$E' = E_A\cup E_B$. We define the weight of every vertex as one:
$w': V' \to \{ 1\}$.

\smallskip

We now show that  $G$ has an independent set $I$ of
size $|I| \ge k$ if and only if  $(k+3, G', w')$ has a pure Nash equilibrium.
We first verify the sufficiency.

\begin{lemma}
If there exists an independent set $I \subseteq V$ satisfying
$k \le |I|$, then $(k+3, G', w')$ has a Nash equilibrium.
\end{lemma}
\begin{proof}
 Let $I' \subseteq I$ be an arbitrary fixed subset of $I$ such that
 $|I'| = k$, and let $v_1, v_2, \dots , v_k$ be the vertices in $I'$.
  Then
 the following strategy profile $\vec{s}^*$ is shown to be a Nash equilibrium.
 
 \smallskip
 \noindent
 {\bf The strategy profile $\vec{s}^*$}: For each $i\in [k]$,
 $p_i$ chooses $v_i$. The player $p_{k+1}$ chooses $b$.
 The players $p_{k+2}$ and $p_{k+3}$ choose $a_2$ and $a_3$, respectively.
 \smallskip

Let us estimate the utility of each player.
Firstly, consider $p_i$ for every $i\in [k]$.
The player $p_i$ chooses $v_i$, and $v_i$ is connected to $n-\delta_v$ vertices
in $D_v$ and $\delta_v$ subdivisions. Since $I'$ is
an independent set, $p_i$ dominates all these vertices.
However, at time two, $p_{k+1}$ dominates every vertex in $V\backslash I'$.
Thus, for every $i \in [k]$,
\begin{eqnarray}\label{eq:unweight_U_1_k}
U_i(\vec{s}^*) = n+1.
\end{eqnarray}
Secondly, consider the player $p_{k+1}$. Since $p_{k+1}$ dominates at least
$b$ and $b_1, b_2, \dots , b_{\bw}$,
\begin{eqnarray}\label{eq:unweight_U_k+1}
U_{k+1}(\vec{s}^*) \ge \bw + 1.
\end{eqnarray}
Lastly, we have
\begin{eqnarray}\label{eq:unweight_U_k+2_k+3}
U_{k+2}(\vec{s}^*) = U_{k+3}(\vec{s}^*) = n+2.
\end{eqnarray}

Below we show that $\vec{s}^*$ is a Nash equilibrium
by verifying Eq.~(\ref{eq:condition_on_Nash}).
We omit the cases where a player changes the strategy to
a vertex chosen by another player.

\smallskip
\noindent
(i) $p_1, p_2, \dots , p_k$

Consider an arbitrary player $p_i$, $i \in [k]$. 
If $p_i$ chooses a vertex in $V_A$, then, clearly $p_i$ dominates
at most $n+1$ vertices.
If $p_i$ chooses $b_e$ for some $e \in E$ or $u \in D_v$ for some $v$, then 
$p_i$ dominates exactly one vertex due to $p_{k+1}$.
Suppose $p_i$ chooses $v' \in V$.
If there exists $e = (u, v')$ such that $u \in I'$,
$p_i$ does not dominate a vertex $b_e$, and hence
$p_i$ dominates less than $n+1$ vertices. Otherwise,
$p_i$ dominates exactly $n+1$ vertices, as in Eq.~(\ref{eq:unweight_U_1_k}).
Thus, the utility of $p_i$ does not increase from Eq.~(\ref{eq:unweight_U_1_k}).

\smallskip
\noindent
(ii) $p_{k+1}$

If $p_{k+1}$ chooses any vertex in $V_A$, $p_{k+1}$ dominates
at most $n+1$ vertices.
If $p_{k+1}$ chooses any vertex in $V_B\backslash \{b \}$, $p_{k+1}$ cannot
dominate $b$, and hence dominates at most
\begin{eqnarray*}
|V_B\backslash \{b, b_1, b_2, \dots , b_{\bw} \}| = n+m+\sum_{v\in V} (n-\delta_v) = n^2+ n - m.
\end{eqnarray*}
Since $\bw = \Theta(n^3)$,
 the utility of $p_{k+1}$ does not increase from Eq.~(\ref{eq:unweight_U_k+1}).

\smallskip
\noindent
(iii) $p_{k+2}, p_{k+3}$

Both of $p_{k+2}$ and $p_{k+3}$ dominate at most $n+1$ vertices, as same
in (i).
Thus, the utilities of $p_{k+2}$ and $p_{k+3}$ do not increase from Eq.~(\ref{eq:unweight_U_k+2_k+3}).
\end{proof}

We can also verify the necessity which complete the proof of the theorem.

\begin{lemma}\label{lem:UW_Leftarrow}
If $(k+3, G', w')$ has a Nash equilibrium, then
we can construct an independent set $I$ of $G$ such that $|I| = k$.
\end{lemma}
\begin{proof}
Let $\vec{s}^*$ be an arbitrary  Nash equilibrium of $(k+3, G', w')$.
We say that a strategy profile $s$ is \emph{standard} if
$s$ satisfies the following conditions:
\begin{description}
\item[(A.1)] $k$ players choose different $k$ vertices in $V \subseteq V'$.
\item[(A.2)] A player choose $b \in V'$. 
\item[(A.3)] Two players choose $a_2, a_3 \in V'$.
\end{description}
\noindent
We then have the following proposition whose proof will be given in the appendix.
\begin{proposition}\label{lem:UW_standard}
If $(k+3, G', w')$ has a Nash equilibrium $\vec{s}^*$, then $\vec{s}^*$ is standard. 
\end{proposition}

Proposition~\ref{lem:UW_standard} implies that $\vec{s}^*$ is standard.
Thus, without loss of generality, we can assume that
the player $p_i$ chooses $v_i \in V$ for every $i\in [k]$, $p_{k+1}$ does
$b$, $p_{k+2}$ does $a_2$, and $p_{k+3}$ does $a_3$.
Then, we can show that
$I = \{v_1, v_2, \dots , v_k \}$ is an independent set, as follows.

Suppose for the sake of contradiction that
there exists an edge $e = (v_i, v_{i'}) \in E$ for $i, i' \in [k]$.
Since $b_e$ becomes neutral at time two, neither $p_i$ nor $p_{i'}$ 
dominate $b_e$. Moreover, $p_{k+1}$ dominates every vertex $v \in V\backslash I'$.
Thus, we have
$U_i(\vec{s}^*) \le n$ and $U_{i'}(\vec{s}^*) \le n$.
On the other hand, if the players $p_i$ and $p_{i'}$ change their strategy to
$a_1$ (or $a_2$), they dominate $n+1$ vertices:
$U_i(\vec{s}^*_{-i}, a_1) = n+1$ and $U_{i'}(\vec{s}^*_{-i}, a_1) = n+1$,
which contradict that $\vec{s}^*$ is a Nash equilibrium.
\end{proof}

\subsection{Nonnegative Integer Weights}\label{ssec:Nonnegative}

In this section, we consider graphs with
nonnegative integer weights, and prove the following theorem.

\begin{theorem}\label{thm:Positive}
{\CD} is NP-complete even for series-parallel graphs with nonnegative
integer weights.
\end{theorem}

It is easy to observe that {\CD} is contained in NP:
Given the number $k$ of players, graph $G$, weight $w$ and
a strategy profile $\vec{s}$, we can compute $U_i(\vec{s})$
for every $i$, $1 \le i \le k$, in polynomial time just by simulating the game
$(k, G, w)$.
Thus,
we can verify in polynomial time
whether $\vec{s}^*$ is a Nash equilibrium by checking Eq.~(\ref{eq:condition_on_Nash})
for every pair of $i$, $1 \le i \le k$, and $u'\in V$.

We below prove that {\CD} is NP-hard by a reduction from {\sc Partition}
which asks, given a multiset $S=\{ s_1, s_2, \dots , s_{2n}\}$ of $2n$ positive integers, whether there exists a subset
$S' \subseteq [2n]$ such that
\begin{eqnarray}\label{eq:partition}
\alpha = \sum_{j\in S'} s_j = \sum_{j\in [2n] \backslash S'} s_j
\end{eqnarray}
where $\alpha = (\sum_{j=1}^{2n} s_j)/2$.
{\sc Partition} is well-known to be NP-complete~\cite{GareyJ79}.

We below construct a 
series-parallel graph $G_S$ and a weight function
$w_S$ for $G_S$ such that there exists $S'\subseteq [2n]$
satisfying Eq.~(\ref{eq:partition})
if and only if $(2n+4, G_S, w_S)$ has a Nash equilibrium.

\begin{figure}[t]
\begin{center}
\includegraphics[width=118mm,clip]{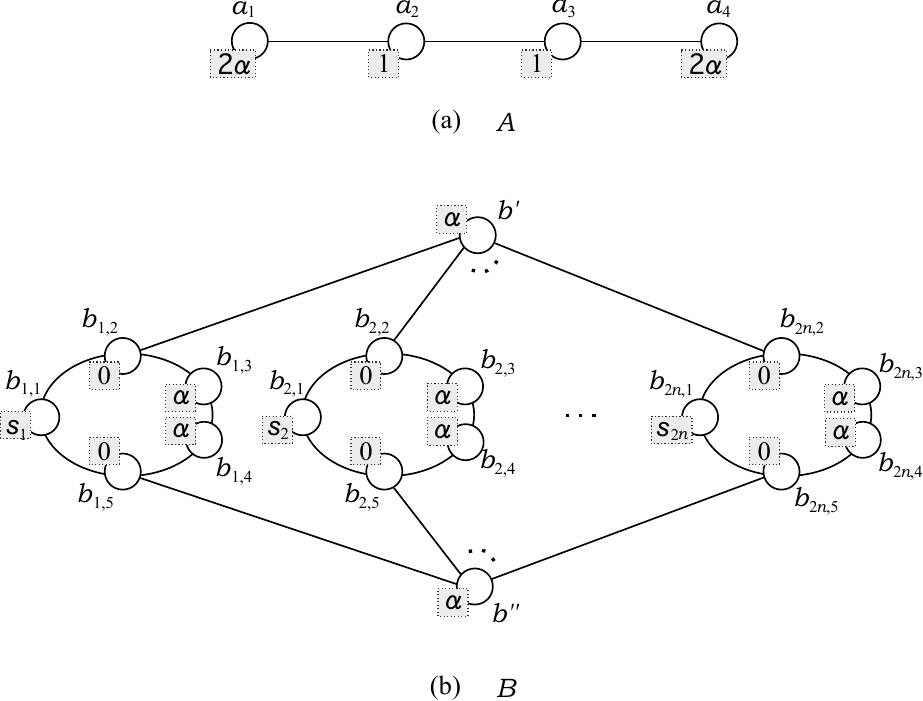}
\caption{Graph $G_S$ for nonnegative integer weights. (a) Graph $A$. (b) Graph $B$.}
\label{fig:positive}
\end{center}
\end{figure}
\smallskip\noindent
[Construction of $G_S$ and $w_S$]

The graph $G_S$ consists of two disjoint component graphs $A=(V_A, E_A)$
and $B=(V_B, E_B)$.

The graph $A$ is a path of four vertices $a_1, a_2, a_3, a_4$, where
for every $i$, $1 \le i \le 4$,
$w_S(a_1) = 2\ps, w_S(a_2) = 1, w_S(a_3) = 1, w_S(a_4) = 2\ps$.
See Fig.~\ref{fig:positive}(a).

We construct $B$, as follows. For every $j \in [2n]$,
we make a cycle of five vertices $b_{j, 1}, b_{j, 2}, b_{j, 3}, b_{j,4},b_{j,5}$
placed in clockwise order.
We complete $B$ by adding two vertices $b'$ and $b''$, and
connecting $b'$ to $b_{j, 2}$ while
connecting  $b''$ to $b_{j, 5}$ for every $j \in [2n]$.
The weights of the vertices in $V_B$ are defined as follows:
$w_S(b_{j, 1}) = s_j, \quad w_S(b_{j, 2}) = 0, \quad w_S(b_{j, 3}) = \ps,
\quad w_S(b_{j, 4}) = \ps, \quad w_S(b_{j, 5}) = 0$;
and
$w_S(b') = w_S(b'') = \ps$. See Fig.~\ref{fig:positive}(b).

We now verify the following two lemmas.
\begin{lemma}
If there exists $S'\subseteq [2n]$
satisfying Eq.~(\ref{eq:partition}), then
$(2n+4, G_S, w_S)$ has a Nash equilibrium.
\end{lemma}
\begin{proof}
Let $S' \subseteq [2n]$ be a subset satisfying
Eq.~(\ref{eq:partition}), and let $S'' = [2n]\backslash S'$.
Note that $|S'|+|S''|= 2n$.
Then, the following strategy profile $\vec{s}^*$ is shown to be
a Nash equilibrium.

\smallskip
\noindent
{\bf The strategy profile $\vec{s}^*$}:
For each $i \in [2n]$, $p_i$ chooses $b_{i, 4}$ if
$i \in S'$; and does $b_{i, 3}$ if $i \in S''$.
The player $p_{2n+1}$ chooses $b'$,
$p_{2n+2}$ does $b''$, $p_{2n+3}$ does $a_2$, and
$p_{2n+4}$ does $a_3$.

\smallskip

We estimate the utility of each player.
Consider an arbitrary player $p_i$, $i \in [2n]$.
Since, for every $j \in [2n]$, $b_{j, 2}$ and $b_{j, 5}$ either
become neutral or is dominated by $p_{2n+1}$ or $p_{2n+2}$,
we have, for every $i \in [2n]$,
\begin{eqnarray}\label{eq:positive_Ui}
U_i(\vec{s}^*)= 2\ps.
\end{eqnarray}
The player $p_{2n+1}$ dominates $b'$ together with $b_{j', 1}$ and $b_{j', 2}$
for every $j' \in S'$. Thus,
\begin{eqnarray}\label{eq:positive_U2n+1}
U_{2n+1}(\vec{s}^*)=\ps + \sum_{j' \in S'} s_{j'} = 2\ps.
\end{eqnarray}
Similarly, $p_{2n+2}$ dominates $b''$ together with $b_{j'', 1}$ and $b_{j'', 5}$
for every $j'' \in S''$. Thus,
\begin{eqnarray}\label{eq:positive_U2n+2}
U_{2n+2}(\vec{s}^*)=\ps + \sum_{j'' \in S''} s_{j''} = 2\ps.
\end{eqnarray}
Each of $p_{2n+3}$ and $p_{2n+4}$ dominate
the two vertices in $A$, and hence
\begin{eqnarray}\label{eq:positive_U2n+3_2n+4}
U_{2n+3}(\vec{s}^*)= U_{2n+4}(\vec{s}^*) = 2\ps+1.
\end{eqnarray}

We now verify that $\vec{s}$ is a Nash equilibrium.
We omit the cases where a player changes the strategy to a vertex
already chosen by another player.

\smallskip
\noindent
(i) $p_1, p_2, \dots , p_{2n}, p_{2n+3}, p_{2n+4}$

Consider an arbitrary player $p_i$, $i \in [2n]$.
Clearly, $U_i(\vec{s}^*_{-i}, v') \le 2\ps$ for every $v' \in V_A$.
If $p_i$ changes the strategy to 
$v' \in \{b_{j,1}, b_{j,2}, b_{j,3}, b_{j,4}, b_{j,5}\}$ for some $j \in [2n]
\backslash \{i \}$, then $U_i(\vec{s}^*_{-i}, v') \le \ps < 2\ps$
due to $p_j, p_{2n+1}$ and $p_{2n+2}$; and
if $p_i$ changes the strategy to 
$v' \in \{b_{i, 1}, b_{i, 2}, b_{i, 3}, b_{i, 4}, b_{i,5}\}$, then
$U_i(\vec{s}^*_{-i}, v') \le 2\ps$
due to $p_{2n+1}$ and $p_{2n+2}$, as desired.

We can similarly verify the case for $p_{2n+3}$ and $p_{2n+4}$.

\smallskip
\noindent
(ii) $p_{2n+1}$ and $p_{2n+2}$

Let $i \in \{ 2n+1, 2n+2\}$. Smilarly to (i),
$U_i(\vec{s}^*_{-i}, v') \le 2\ps$ for every $v' \in V_A$.
If $p_i$ changes the strategy to 
$v' \in \{b_{j,1}, b_{j,2}, b_{j,3}, b_{j,4}, b_{j,5}\}$ for some $j \in [2n]$,
then $U_i(\vec{s}^*_{-i}, v') \le s_j + \ps \le 2\ps$, as desired.
\end{proof}

\begin{lemma}\label{lem:Nonnegative_Leftarrow}
If $(2n+4, G_S, w_S)$ has a Nash equilibrium,
then there exists $S'\subseteq [2n]$
satisfying Eq.~(\ref{eq:partition}).
\end{lemma}
\begin{proof}
Let $\vec{s}^*$ be an arbitrary  Nash equilibrium of $(2n+4, G_S, w_S)$.
We say that a strategy profile $s$ is \emph{standard} if
$s$ satisfies the following conditions:
\begin{description}
\item[(B.1)] For every $i\in [2n]$, exactly one player
chooses $b_{i, 3}$ or $b_{i,4}$; and no player chooses
$b_{i, 1}, b_{i,2}$ or $b_{i,5}$.
\item[(B.2)] Exactly one player chooses $b' \in V_B$,
and exactly one player does $b''$.
\item[(B.3)]  Exactly one player chooses $a_2 \in V_A$,
and exactly one player does $a_3$.
\end{description}
\noindent
We then have the following proposition whose proof will be given in the appendix.
\begin{proposition}\label{lem:Nonnegative_standard}
If $(2n+4, G_S, w_S)$ has a Nash equilibrium $\vec{s}^*$, then
$\vec{s}^*$ is standard. 
\end{proposition}

Proposition~\ref{lem:Nonnegative_standard} implies that $\vec{s}^*$
is standard.
Let $S'$ (resp., $S''$) be the set of indices $i \in [2n]$ which a player
chooses $b_{i,4}$ (resp., $b_{i, 3}$).
Let $p_{2n+1}$ and $p_{2n+2}$ be the players choosing $b'$ and $b''$,
respectively.
Clearly,
\begin{eqnarray}\label{eq:positive_S'_1}
U_{2n+1}(\vec{s}^*) = \ps + \sum_{j'\in S'}s_{j'},
\end{eqnarray}
while
\begin{eqnarray}\label{eq:positive_S'_2}
U_{2n+1}(\vec{s}^*_{-(2n+1)}, a_1) = 2\ps.
\end{eqnarray}
Since $\vec{s}^*$ is a Nash equilibrium, Eqs.~(\ref{eq:positive_S'_1}) and
(\ref{eq:positive_S'_2}) imply that
$2\ps \le \ps + \sum_{j'\in S'} s_{j'}$,
and hence
\begin{eqnarray}\label{eq:positive_S'}
\ps \le \sum_{j'\in S'}s_{j'}.
\end{eqnarray}
Similarly, we have
$U_{2n+2}(\vec{s}^*) = \ps + \sum_{j''\in S'}s_{j''}$,
and hence
\begin{eqnarray}\label{eq:positive_S''}
\ps \le \sum_{j'\in S''}s_{j''}.
\end{eqnarray}
Since $\sum_{j'\in S'}s_{j'} + \sum_{j''\in S''} s_{j''} = 2\ps$,
Eqs.~(\ref{eq:positive_S'}) and (\ref{eq:positive_S''}) imply that
$S'$ and $S''$ are the desired partition.
\end{proof}

\subsection{Arbitrary Integer Weights}\label{ssec:NP_Arbitrary}

In this section, we prove that
the problem is NP-hard even for forests.

\begin{theorem}\label{thm:Arbitrary}
{\sc Competitive Diffusion} is NP-complete even for forests of two components.
\end{theorem}

To prove the theorem, we employ {\sc Partition}, as in the proof of
Theorem~\ref{thm:Positive}; but
we here assume that
the subset $S'$ satisfies $|S'|=n$;
 {\sc Partition} remains NP-complete under this additional condition~\cite{GareyJ79}.
Moreover, We assume without loss of generality that
$s_1, s_2, \dots, s_{2n}$ are even; if it is not the case, we simply
double each of $s_1, s_2, \dots, s_{2n}$.

Let $S=\{ s_1, s_2, \dots , s_{2n}\}$ be an input of {\sc Partition},
and $\alpha = (\sum_{j=1}^{2n} s_j)/2$.
We show that one can construct a forest $G_S$ and a weight function
$w_S$ for $G_S$ such that there exists $S'\subseteq [2n]$
that satisfies $|S'|=n$  and Eq.~(\ref{eq:partition})
if and only if $(2n+4, G_S, w_S)$ has a Nash equilibrium.
The graph $G_S$ consists of two trees $A=(V_A, E_A)$ and $B = (V_B, E_B)$.

\begin{figure}[t]
\begin{center}
\includegraphics[width=118mm]{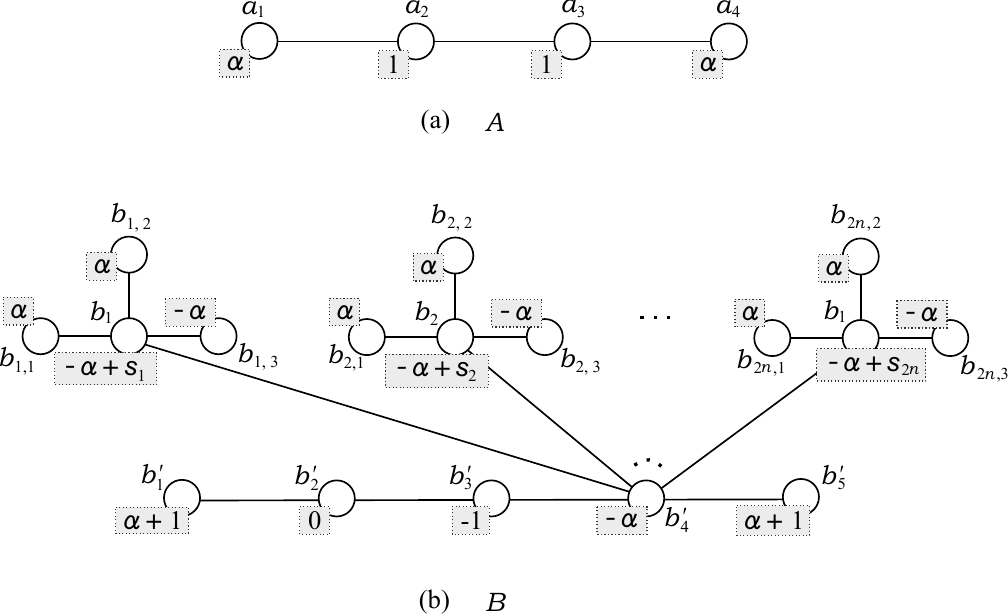}
\caption{Graph $G_S$ for arbitrary integer weights. (a) Graph $A$. (b) Graph $B$. }
\label{fig:forest_NPhard}
\end{center}
\end{figure}
\smallskip
\noindent
[Construction of $G_S$ and $w_S$]

The graph $A$ is a path of four vertices $a_1, a_2, a_3, a_4$, where
for every $i$, $1 \le i \le 4$,
$w_S(a_1) = \ps, w_S(a_2) = 1, w_S(a_3) = 1, w_S(a_4) = \ps$.
See Fig.~\ref{fig:forest_NPhard}(a).

We then construct $B$, as follows. For every $j \in [2n]$,
we make a graph $B_j$ consisting of a vertex $b_j$ together with
$b_{j,1}, b_{j,2}$ and $b_{j,3}$
which compose a complete bipartite graph $K_{1, 3}$.
We complete $B$ by making a path of five vertices $b_1', b_2', b_3', b_4', b_5'$,
and connecting $b_4'$ to $b_j$ for every $j \in [2n]$.
The weights of the vertices in $V_B$ are defined as follows:
For every $j\in [2n]$,
$w_S(\ba) = -\ps+s_j, w_S(\bb) = \ps, w_S(\bc) = \ps, w_S(\bd) = -\ps$;
and
$w_S(b'_1) = \ps + 1, w_S(b'_2) = 0, w_S(b'_3) =  -1,
 w_S(b'_4) = -\ps, w_S(b'_5) =
\ps + 1$.
See Fig.~\ref{fig:forest_NPhard}(b).

Similarly to the previous sections, we can prove that
there exists $S'\subseteq [2n]$
satisfying Eq.~(\ref{eq:partition})
if and only if $(2n+4, G_S, w_S)$ has a Nash equilibrium.
The proof is omitted, but given in the 
the appendix.

\section{Algorithms for Forests of Paths}\label{sec:AlgForPaths}
 
In the last section, we have shown that {\CD} is a computationally hard problem.
However, we can solve the problem for some particular graph classes.
In Section~\ref{ssec:FofP_weighted}, we give a pseudo-polynomial-time algorithm to solve 
{\CD} for forests of weighted paths; 
as its consequence, we show that the problem is solvable in polynomial time for forest of unweighted paths.
In Section~\ref{ssec:FofP}, we improve the running time of our algorithm to a quadratic time for the unweighted case.

\subsection{Forests of weighted paths}\label{ssec:FofP_weighted}

Let $F$ be a forest consisting of weighted $m$ paths $P_1, P_2, \ldots, P_m$, and let $W_j$ be the sum of the positive weights in a path $P_j$, $j \in [m]$. 
Then, we define $W = \max_{j \in [m]} W_j$ as the {\em upper bound on utility} for $F$, that is, any player can obtain at most $W$ in $F$. 
In this subsection, we prove the following theorem. 
\begin{theorem} \label{the:path} 
Let $F$ be a forest of weighted paths. Let $n$ and $W$ be the number of vertices in $F$
and the upper bound on utility for $F$, respectively. Then, we can solve
{\CD}, and find a Nash equilibrium, if any, in $O(W n^9)$ time. 
\end{theorem}

We note that $W=O(n)$ if $F$ is an unweighted graph.
Therefore, by Theorem~\ref{the:path}, {\CD} is solvable in $O(n^{10})$ time for an unweighted graph $F$; 
this running time will be improved to $O(n^2)$ in Section~\ref{ssec:FofP}. 

\smallskip

\noindent
{\bf Idea and definitions.}

Let $F$ be a given forest consisting of weighted $m$ paths $P_1, P_2, \dots , P_m$.
Let $w$ be a given weight function; we sometimes denote by $w_j$ the weight function restricted to the path $P_j$, $j \in [m]$. 
Suppose that, for an integer $k$, there exists a strategy profile $\vec{s}$ for the game $(k, F, w)$ that is a Nash equilibrium.
Then, the strategy profile restricted to each path $P_j$, $j \in [m]$, forms a Nash equilibrium for $(k_j, P_j, w_j)$, where $k_j$ is the number of players who chose vertices in $P_j$. 
However, the other direction does not always hold:
A Nash equilibrium $\vec{s}_j$ for $(k_j, P_j, w_j)$ is not always extended to a Nash equilibrium for the whole forest $F$,
because some player may increase its utility by moving to another path in $F$. 
To capture such a situation, we classify a Nash equilibrium for a (single) path $P_j$ more precisely. 




Consider the game $(\kappa_j, P_j, w_j)$ for an integer $\kappa_j \ge 0$.
For a strategy profile $\vec{s}_j$ for $(\kappa_j, P_j, w_j)$, 
we define $\umin_{P_j}(\vec{s}_j)$ as the minimum utility over all the $\kappa_j$ players:
$\umin_{P_j}(\vec{s}_j) = \min_{i\in [\kappa_j]} U_i(\vec{s}_j)$.
In other words, any player in $P_j$ obtains the utility at least
$\umin_{P_j}(\vec{s}_j)$.
For the case where $\kappa_j=0$, we define $\vec{s}_j = \emptyset$ as
the unique strategy profile for $(\kappa_j, P_j, w_j)$; 
then, $\vec{s}_j$ is a Nash equilibrium and we define $\umin_{P_j} (\vec{s}_j) = + \infty$.

For a strategy profile $\vec{s}_j = \bigl( s_j^{(1)}, s_j^{(2)}, \ldots, s_j^{(\kappa_j)} \bigr)$ for $(\kappa_j, P_j, w_j)$, we then define the ``potential'' of the maximum utility under $\vec{s}_j$ that can be expected to gain by an extra player other than the $\kappa_j$ players.
More formally, for a vertex $v$ in $P_j$, we denote by $\vec{s}_j + v$ the strategy profile $\bigl( s_j^{(1)}, s_j^{(2)}, \ldots, s_j^{(\kappa_j)}, s_j^{(\kappa_j+1)} \bigr)$ for $(\kappa_j+1, P_j, w_j)$ such that $s_j^{(\kappa_j+1)} = v$.  
Then, we define $\suki_{P_j}(\vec{s}_j) = \max_{v \in V(P_j)} U_{\kappa_j+1}(\vec{s}_j+v)$.

For two nonnegative integers $\kappa_j$ and $t$, 
we say that \emph{$P_j$ admits $\kappa_j$ players with a boundary $t$} if
there exists a strategy profile $\vec{s}_j$ such that
$\vec{s}_j$ is a Nash equilibrium for $(\kappa_j, P_j, w_j)$ and
$\suki_{P_j}(\vec{s}_j) \le t \le \umin_{P_j}(\vec{s}_j)$ holds.
Then, the following lemma characterizes a Nash equilibrium
of the game $(k, F, w)$ in terms of the components of $F$; the proof is given in the appendix.
\begin{lemma}\label{lem:decompose_players}
The game $(k, F, w)$ has a Nash equilibrium if and only if
there exist nonnegative integers $\kappa_1, \kappa_2, \dots , \kappa_m$ and $t$ such that
$k=\sum_{j=1}^m \kappa_j$ and $P_j$ admits $\kappa_j$ players with the common boundary $t$
for every $j \in [m]$.
\end{lemma}

\smallskip
\noindent
{\bf Algorithm.}
 
We first focus on a weighted single path. 
\begin{lemma}\label{lem:ForestsOfPaths_weighted}
Let $P$ be a weighted path of $n$ vertices, and
$t$ be a nonnegative integer.
Then, one can find in $O(n^9)$ time the set $\akianspset \subseteq \{ 0, 1, \ldots , 2n \}$ of all the integers $\kappa$ such that $P$ admits $\kappa$ players
with boundary $t$.
\end{lemma}


Based on Lemma~\ref{lem:ForestsOfPaths_weighted}, we can obtain the $m$ sets $\akianspset_1, \akianspset_2, \ldots, \akianspset_m \subseteq \{ 0, 1, \ldots , 2n \}$, where $\akianspset_j \subseteq \{ 0, 1, \ldots , 2n \}$, $j \in [m]$,
is
the set of all the integers $\kappa$ such that $P$ admits $\kappa$ players with boundary $t$.
This can be done in $O(n^9)$ time, where $n$ is the number of vertices in the whole forest $F$. 

We now claim that, for a given integer $t$, it can be decided in $O(n^3)$ time whether there exist nonnegative integers $\kappa_1, \kappa_2, \dots , \kappa_m$ such that
$k=\sum_{j=1}^m \kappa_j$ and $P_j$ admits $\kappa_j$ players with the common boundary $t$ for every $j \in [m]$;
later we will apply this procedure to all possible values of $t$, $0 \le t \le W$. 
To show this, observe that finding desired $m$ integers $\kappa_1, \kappa_2, \dots , \kappa_m$ from the $m$ sets $\akianspset_1, \akianspset_2, \ldots, \akianspset_m$ can be regarded as solving an instance of the multiple-choice knapsack problem~\cite{KPP}:
The capacity $c$ of the knapsack is equal to $k$; 
Each integer $\kappa^\prime$ in $\akianspset_j$, $j \in [m]$, corresponds to an item with profit $\kappa^\prime$ and cost $\kappa^\prime$; 
The items from the same set $\akianspset_j$ form one class, from which at most one item can be packed into the knapsack. 
The multiple-choice knapsack problem can be solved in $O(cN)$ time~\cite{KPP}, where $N$ is the number of all items. 
Since $c = k$ and $N = O(mn)$, we can solve the corresponding instance in time $O(kmn) = O(n^3)$. 

We finally apply the procedure above to all possible values of boundaries $t$. 
Since any player can obtain at most the upper bound $W$ on utility for $F$, 
it suffices to consider $t \in [W]$. 
Therefore, our algorithm runs in $O(Wn^9)$ time in total.



\subsection{Forests of unweighted paths}\label{ssec:FofP}

In this subsection, we improve the running time of our algorithm in Section~\ref{ssec:FofP_weighted} to a quadratic time when restricted to the unweighted case. 
\begin{theorem}
Let $F$ be a forest of unweighted paths, and $n$ be the number of vertices in $F$.
Then, we can solve {\CD}, and find a Nash equilibrium, if any, 
in $O(n^2)$ time.
\end{theorem}

In the rest of this subsection, we consider unweighted graphs, and thus
define $w: V \to \{ 1\}$ for the vertex set $V$ of a given forest.
We assume that the number $k$ of players is less than $n$; otherwise,
a Nash equilibrium always exists. Note that, in this case,
every player has utility at least one for any Nash equilibrium. 

We first show that the set $\akianspset_j$ of Lemma~\ref{lem:ForestsOfPaths_weighted} can be obtained in $O(1)$ time, instead of $O(n^9)$ time,  
%
%
%
%
by characterizing Nash equilibriums
for $(\kappa, P, w)$ in terms of $\kappa, t$ and $n$.
\begin{lemma}\label{lem:ForestsOfPaths}
Let $P$ be a single unweighted path of $n$ vertices,
and let $\kappa$ and $t$ be nonnegative and positive integers, respectively.
\begin{description}
\item[(1)] $P$ admits $\kappa = 0$ player with $t$ if and only if $n \le t$.
\item[(2)] $P$ admits $\kappa=1$ player with $t$ if and only if
$t \le n \le 2t+1$.
\item[(3)] $P$ admits $\kappa=2$ players with $t$ if and only if
$2t \le n \le 2t+2$.
\item[(4)] $P$ admits $\kappa=3$ players with $t$ if and only if
$t = 1$ and $n = 3, 4$ or $5$.
\item[(5)] For any integer $\kappa \ge 4$, $P$ admits $\kappa$ players with $t$ if and only if
\begin{eqnarray}\label{eq:FofP_k >= 4}
\begin{array}{ll}
(\kappa+1)t-1 \le n \le (2\kappa-4)t+\kappa & \mbox{ if $\kappa$ is odd};\\
\kappa t \le n \le (2\kappa-4)t+\kappa & \mbox{ if $\kappa$ is even}.\\
\end{array}
\end{eqnarray}
\end{description}
\end{lemma}

By Lemme~\ref{lem:ForestsOfPaths}, we can immediately
obtain the number of players which $P$ admits with a given boundary
$t$:

\begin{corollary}\label{coro:FofP}
Consider a fixed boundary $t$.
If $P$ is a path of $n$ vertices,
the numbers of players which $P$ admits with a boundary
$t$ is given as follows.
\begin{description}
\item[(1)] If  $n \le t-1$, the number is only $0$.
\item[(2)] If $n=t$, the numbers are $0$ and $1$.
\item[(3)] If $t+1 \le n \le 2t-1$, the number is only $1$.
\item[(4)] If $2t \le n \le 2t+1$ and $n=3$, the numbers are $1, 2$ and $3$; and if $2t \le n \le 2t+1$ and $n\ne3$, the numbers are $1$ and $2$.
\item[(5)] If $n = 2t+2$ and $n=4$, the numbers are $2, 3$ and $4$; and if $n = 2t+2$ and $n\ne 4$, the number is only $2$.
\item[(6)] If $2t+3 \le n \le 4t-1$, $P$ has no desired Nash equilibrium.
\item[(7)] If $4t \le n$ and $5 \le n$,
the numbers are integers $\kappa$ such that
\[
\left\lceil \frac{n+4t}{2t+1} \right\rceil \le \kappa \le \max(k_{odd}, k_{even}),
\]
where $k_{odd}$ is the maximum odd integer satisfying
$k_{odd} \le (n-t+1)/t$,
and $k_{even}$ is the maximum even integer satisfying
$k_{even} \le n/t$.
\end{description}
\end{corollary}
We use Corollary~\ref{coro:FofP} to design our algorithm for forests of paths.

Without loss of generality, we assume that
$P_1$ is a longest path among the $m$ paths, and has $n_1$ vertices.
For each $t$, $1 \le t \le n_1$, we repeat the following procedure:
For every $j$, $1 \le j \le m$, we obtain, using Corollary~\ref{coro:FofP},
the minimum number $k^{\min}_j$ and the maximum number $k^{\max}_j$ of players
which $P_j$ admits with the boundary $t$.
Corollary~\ref{coro:FofP} implies that, for every $j$, $1 \le j \le m$,
$P_j$ admits $\kappa$ players with $t$
for any $\kappa$ between $k^{\min}_j$ and $k^{\max}_j$, and hence
$(k, F, w)$ has a Nash equilibrium with the common boundary $t$
if and only if
\begin{eqnarray}\label{eq:FofP}
\sum_{j=1}^m k^{\min}_j \le k \le \sum_{j=1}^{m} k^{\max}_j.
\end{eqnarray}
We thus complete the procedure by checking if the two inequalities in~(\ref{eq:FofP}) both hold. Since Corollary~\ref{coro:FofP} implies that
we can obtain $k^{\min}_j$ and $k^{\max}_j$ in constant time for every $j$,
the running time of the procedure above for single $t$ is $O(m)$, and hence
that of our entire algorithm is $O(n_1m)=O(n^2)$, as desired.

\section{Algorithms for Chain, Cochain, and Threshold Graphs}\label{sec:AlgForCochain}

A bipartite graph $B = (X, Y; E)$ with $|X| = p$ and $|Y| = q$
is a \emph{chain graph} if there is an ordering
$(x_{1}, x_{2}, \dots, x_{p})$ on $X$ such that
$N(x_{1}) \subseteq N(x_{2}) \subseteq \cdots \subseteq N(x_{p})$,
where $N(u)$ denote a set of neighbors of a vertex $u$.
If there is such an ordering on $X$,
then there also exists an ordering $(y_{1}, y_{2}, \dots, y_{q})$ on $Y$
such that $N(y_{1}) \subseteq N(y_{2}) \subseteq \cdots \subseteq N(y_{q})$.
We call such orderings \emph{inclusion orderings}.
A graph $B'$ is a \emph{cochain graph}
if it can be obtained from a chain graph $B = (X,Y;E)$
by making the independents sets $X$ and $Y$ into cliques.
A graph $B''$ is a \emph{threshold graph}
if it can be obtained from a chain graph $B = (X,Y;E)$
by making one of the independents sets $X$ and $Y$ into a clique.
Observe that inclusion orderings on $X$ and $Y$ in $B$
can be seen as \emph{inclusion orderings} in $B'$ and $B''$
if we use closed neighborhoods in cliques.
Such inclusion orderings can be found in linear time~\cite{Heggernes:08}.
Because the algorithm for chain graphs we will describe in this section
depends only on its property of having inclusion orderings,
we can apply the exactly same algorithm for cochain graphs and threshold graphs.

The following lemma follows directly from the definitions. Note that $N[u] = N(u)
\cup \{ u \}$.
\begin{lemma}
\label{lem:inclusion}
If $N(u) \subseteq N(v)$ or $N[u] \subseteq N[v]$ holds for $u = s^{(i)} \ne v = s^{(j)}$,
then
\[
  U_{i}(\vec{s}) =
  \begin{cases}
    0 & \text{if there is } h \ne i \text{ such that } s^{(h)} = u, \\
    w(u) & \text{otherwise}.
  \end{cases}
\]
\end{lemma}

In what follows, 
let $B = (X,Y; E)$ be a chain graph with inclusion orderings
$(x_{1}, \dots, x_{p})$ and $(y_{1}, \dots, y_{q})$ on $X$ and $Y$, respectively.
We define $\sMax{X} = \max(\{0\} \cup \{i \mid x_{i} \in V(\vec{s})\})$
and $\sMax{Y} = \max(\{0\} \cup \{i \mid y_{i} \in V(\vec{s})\})$.

\begin{lemma}
\label{lem:chain_heaviest}
Let $\vec{s}$ be a Nash equilibrium of $B$.
If $s^{(i)} \notin \{x_{\sMax{X}}, y_{\sMax{Y}}\}$, then
\begin{eqnarray}\label{eq:chain}
  w(s^{(i)}) \ge
  \max \left\{w(u) \mid u \in \bigl( \{x_{j} \mid j \le \sMax{X}\} \cup \{y_{j} \mid j \le \sMax{Y}\}\bigr) \setminus V(\vec{s}) \right\}.
\end{eqnarray}
\end{lemma}
\begin{proof}
Since $N(s^{(i)}) \subseteq N(x_{\sMax{X}})$ or $N(s^{(i)}) \subseteq N(y_{\sMax{Y}})$,
it follows that $U_{i}(\vec{s}) \le w(s^{(i)})$ by Lemma~\ref{lem:inclusion}.
Suppose for the contrary that
there exists $u \in ( \{x_{j} \mid j \le \sMax{X}\} \cup \{y_{j} \mid j \le \sMax{Y}\}) \setminus V(\vec{s})$
such that $w(s^{(i)}) < w(u)$.
Now it holds that $N(u) \subseteq N(x_{\sMax{X}})$ or $N(u) \subseteq N(y_{\sMax{Y}})$.
Thus, by Lemma~\ref{lem:inclusion}, we have $U_{i}(\vec{s}_{-i}, u) = w(u) > w(s^{(i)}) \ge U_{i}(\vec{s})$.
This contradicts the assumption that $\vec{s}$ is a Nash equilibrium.
\end{proof}

Thus, it suffices to check the strategy profiles satisfying Eq.~(\ref{eq:chain})
for our purpose.

\begin{algorithm}[htbp]
\caption{Find a Nash equilibrium $\vec{s} \in V^{k}$ of a chain graph $B = (X,Y; E)$}
\label{alg:chain}
\begin{algorithmic}[1]
  \STATE Let $(x_{1}, \dots, x_{p})$ on $X$ and $(y_{1}, \dots, y_{q})$ on $Y$ be inclusion orderings.
  \STATE // The following is for the case where $\sMax{X} \ne 0$.
  \FORALL{guesses $(\sMax{X}, \sMax{Y}) \in \{1,\dots, p\} \times \{0,\dots,q\}$}
  \STATE $s^{(1)} := x_{\sMax{X}}$. \ $s^{(2)} := y_{\sMax{Y}}$ \textbf{if} $\sMax{Y} \ne 0$.
  \STATE $R := \{x_{i} \mid i < \sMax{X}\} \cup \{y_{i} \mid i < \sMax{Y}\}$.
  \WHILE{there is a player $i$ not assigned to a vertex}
    \STATE $v := \argmax_{u \in R} w(u)$.
    \IF{$w(v) \ge 0$}
      \STATE $s^{(i)} := v$. \ $R := R \setminus \{v\}$.
    \ELSE
      \STATE $s^{(i)} := x_{\sMax{X}}$.
      \label{line:no-positive}
    \ENDIF
  \ENDWHILE
  \STATE \textbf{return} $\vec{s}$ \textbf{if} it is a Nash equilibrium.
  \ENDFOR
  \RETURN ``no Nash equilibrium''
\end{algorithmic}
\end{algorithm}

\begin{theorem}
Let $G$ be a chain, cochain, or threshold graph of $n$ vertices and $m$ edges. Then, we can solve {\CD} for $G$, and find a Nash equilibrium, if any, in $O(n^{4} (m+n))$ time. 
\end{theorem}
\begin{proof}
We present an algorithm for chain graphs only.
As previously described, we can apply the same algorithm for cochain and threshold graph.
 
We first guess $\sMax{X}$ and $\sMax{Y}$. Here we assume $\sMax{X} \ne 0$.
The other case can be treated in the same way by swapping $X$ and $Y$.
We assign $x_{\sMax{X}}$ to the first player.
If $\sMax{Y} \ne 0$, then we assign $y_{\sMax{Y}}$ to the second player.
By Lemma~\ref{lem:chain_heaviest},
if $\vec{s}$ is a Nash equilibrium, then the other players have to select the heaviest vertices
in $\{x_{i} \mid i < \sMax{X}\} \cup \{y_{i} \mid i < \sMax{Y}\}$.
For each of the remaining players, we assign a vacant vertex with the maximum non-negative weight.
If there is no such a vertex, we assign $x_{\sMax{X}}$.
We then test whether the strategy profile is a Nash equilibrium.
See Algorithm~\ref{alg:chain}.

Lemma~\ref{lem:chain_heaviest} implies that
if the algorithm assigns at most one player to $x_{\sMax{X}}$,
then the algorithm is correct.
If two or more players are assigned to $x_{\sMax{X}}$,
then these players have utility 0.
In such a case, there are not enough number of vertices of non-negative weights in
$\{x_{i} \mid i < \sMax{X}\} \cup \{y_{i} \mid i < \sMax{Y}\}$.
Thus every $\vec{s}$ with the guesses $\sMax{X}$ and $\sMax{Y}$ has a player with non-positive utility.
If such a player, say $p_{i}$, has negative utility, then $\vec{s}$ is clearly not a Nash equilibrium.
If $p_{i}$ has utility 0, then it may improve its utility only if there is a vertex 
$v \in \{x_{\sMax{X}+1}, \dots, x_{p}\} \cup \{y_{\sMax{Y}+1}, \dots, y_{q}\}$
such that $U_{i}(\vec{s}_{-i}, v) > 0$.
However, in this case, there is no Nash equilibrium with the guesses $\sMax{X}$ and $\sMax{Y}$.
Therefore, the algorithm is correct.

We now analyze the running time of the algorithm.
We have $O(n^{2})$ options for guessing $x_{\sMax{X}}$ and $y_{\sMax{Y}}$.
For each guess, the bottle-neck of the running time is to test
whether the strategy profile is a Nash equilibrium or not.
It takes $O(n^{2} (m+n))$ time as follows:
we have $O(n^{2})$ candidates of moves of players;
for each candidate, we can compute the utility of the player moved by running a breadth-first search once
in $O(m+n)$ time by adding a virtual root connecting to all the vertices occupied by the players.
In total, the algorithm runs in $O(n^{4} (m+n))$ time.
\end{proof}

\newpage

\appendix

\section{Proof of Proposition~\ref{lem:UW_standard}}

Our proof of Proposition~\ref{lem:UW_standard}
consists of a sequence of claims, and we prove each claim 
by contradiction: we suppose that a claim does not hold,
and show that it contradicts the fact that $\vec{s}^*$ is a Nash
equilibrium.

Let $\vec{s}^*$ be an arbitrary Nash equilibrium of $(k+3, G', w')$.
We start by the following claim.
\begin{claim}\label{clm:UW_AtMostTwo}
There are at most two players choosing vertices in $V_A$.
\end{claim}
\begin{proof}
Suppose for the sake of contradiction that there are $k'(\ge 3)$ players
choosing vertices in $V_A$. Then there exists a players $p_i$ dominating
at most one vertex: $U_i(\vec{s}^*) \le 1$.
On the other hand, every $v \in V$ has
$n- \delta_v (\ge 1)$ neighbors for every $v \in V$, and $b$ has
$\bw + n (\ge 1)$
neighbors, and thus there are $n+1$ vertices (i.e., ones in $V$ and $b$) by which
$p_i$ can dominate at least two vertices by changing the strategy.
Since we have $k' \ge 3$, there are at most $k$($\le n$) players choosing vertices
in $V_B$, and hence we have a contradiction.
\end{proof}

Then, we prove the following claim.
\begin{claim}\label{clm:UW_A.2}
$\vec{s}^*$ satisfies A.2.
\end{claim}
\begin{proof}
Suppose $\vec{s}^*$ does not satisfy A.2, that is,
no player chooses $b$.
Since Claim~\ref{clm:UW_AtMostTwo} implies that at least $k+3 - 2 \ge 2$ players
choose vertices in $V_B$,
there exists a player $p_i$ who cannot dominate $b$.
Since $b_1, b_2, \cdots , b_{\bw}$ are connected only to $b$,
$p_i$ dominates either at most one of $b_1, b_2, \cdots , b_{\bw}$
or at most all the other vertices in $V_B$. Thus,
\begin{eqnarray*}
U_i(\vec{s}) &\le& n + m + \sum_{v \in V}|D_v|\\
&\le&  n + m + n^2 - 2m\\
&\le&  n^2 + n  - m.
\end{eqnarray*}
On the other hand, $p_i$ can dominate, by changing the strategy to $b$,
at least $b$ together with
some of $b_1, b_2, \cdots , b_{\bw}$ and the vertices in $V$.
In this case, since each of the other $k+2 \le n+2$ players can dominate at most
two vertices of $b_1, b_2, \cdots , b_{\bw}$ and the vertices in $V$
due to $p_i$, we have
\begin{eqnarray*}
U_i(\vec{s}_{-i}, b) \ge \bw - 2(n+2) = \Theta(n^3).
\end{eqnarray*}
\end{proof}

We complete the proof of Lemma~\ref{lem:UW_standard} by verifying
the following two claims.
\begin{claim}\label{clm:UW_A.3}
$\vec{s}^*$ satisfies A.3.
\end{claim}
\begin{proof}
By Claim~\ref{clm:UW_AtMostTwo}, it suffices to consider the case where
there are at most two players choosing vertices in $V_A$.

Consider the case where exactly two players choose vertices in $V_A$
other than ``$a_2$ and $a_3$.''
Without loss of generality, the chosen vertices are ``$a_1$ and $a_2$,''
``$a_1$ and $a_3$,'' ``$a_1$ and $a_4$,'' or
one of the players chooses $a_1'$. For any of them, one player
can increase the utility by choosing $a_2$ or $a_3$.

Consider the other case where at most one player chooses a vertex in $V_A$.
Then, there are at least $k+2$ players choosing vertices in $V_B$.
Since Claim~\ref{clm:UW_A.2} implies that some player chooses $b$,
each of the other players dominates at most $n+1$ vertices;
let $p_i$ be one of them: $U_i(\vec{s}^*) \le n+1$.
On the other hand, we have either
\[
U_i(\vec{s}_{-i}, a_2) \ge n+2 \quad \mbox{or} \quad U_i(\vec{s}_{-i}, a_3) \ge n+2.
\]
\end{proof}

\begin{claim}
$\vec{s}^*$ satisfies A.1.
\end{claim}
\begin{proof}
Suppose there exists a player $p_i$ choosing either $u \in D_v$ for some $v \in V$,
$b_e$ for some $e \in E$, or one of $b_1, b_2, \dots , b_{\bw}$.
Claim~\ref{clm:UW_A.2} implies that $p_i$ dominates only one vertex:
$U_i(\vec{s}^*) \le 1$.
However,
Claim~\ref{clm:UW_A.3} implies that
\[
U_i(\vec{s}, a_1) \ge n+1.
\]
\end{proof}

\section{Proof of Proposition~\ref{lem:Nonnegative_standard}}

In this section, we prove Proposition~\ref{lem:Nonnegative_standard}.
As in the last section, our proof consists of a sequence of claims.

Consider an arbitrary Nash equilibrium $\vec{s}^*$ of $(2n+4, G_S, w_S)$.
We first verify the following claim:
\begin{claim}\label{claim:nonnegative_AtMostTwo}
$\vec{s}^*$ satisfies (B.3).
\end{claim}
\begin{proof}
Suppose $\vec{s}^*$ does not satisfy (B.3).
We consider the following three cases (i), (ii) and (iii).

\smallskip
\noindent
(i) \emph{More than three players choose vertices in $A$.}

Suppose there are at least three players choosing vertices in $V_A$, which implies
that
there exists a player $p$ whose utility is at most one. Moreover,
at most $2n + 1$ players choose vertices in $V_B$,
and there are $4n + 2$ vertices with weight $\ps$.
Thus, the player $p$
obtain the utility $\ps (> 1)$ by changing the strategy to a vertex in $V_B$,
which contradicts the fact that $\vec{s}^*$ is a Nash equilibrium.

\smallskip
\noindent
(ii) \emph{Exactly two players choose vertices in $A$.}

Since $\vec{s}^*$ does not satisfy (B.3),
we can assume without loss of generality that the two players
choose either ``$a_1$ and $a_2$,'' ``$a_1$ and $a_3$'' or `` $a_1$ and $a_4$.''
For both cases, a player can increase the utility
by choosing $a_2$ or $a_3$.

\smallskip
\noindent
(iii) \emph{At most one player chooses a vertex in $A$.}

In this case, there exists a vertex in $V_A$ such that
any player obtain utility $2\ps+1$
by changing the strategy to it. Thus,
it suffices to show that one of the other $2n+3$ (or $2n+4$)
players has utility at most $2\ps$.
Clearly, the players, in this case, choose vertices in $V_B$,
while the sum of the weights of all the vertices in $V_B$ is
\[
\sum_{j \in [2n]} (2\ps + s_j) + \ps + \ps = 4n\ps + 4\ps = 2\ps (2n+2).
\]
Thus, at least one player has utility less than $2\ps$, as desired.
\end{proof}

We  now consider the other $2n+2$ players choosing vertices in $V_B$.
For every $j$, $j \in [2n]$, we define $P_j$ as a set of indices of
the players who choose $b_{j, 1},b_{j, 2},b_{j, 3},b_{j, 4}$ or $b_{j, 5}$
on $\vec{s}^*$.
We then have the following claim.

\begin{claim}\label{clm:nonnegative_Pj<=1}
For every $j \in [2n]$, $|P_j| \le 1$.
\end{claim}
\begin{proof}
Note that, by Claim~\ref{claim:nonnegative_AtMostTwo},
any player can obtain the utility $2\ps$ by changing the strategy to $a_1$.
Furthermore, the sum of all the weights of the vertices in $V_B$ is
exactly $2\ps(2n+2)$.
Thus, to derive a contradiction, it suffices to show that
some vertex with a positive weight becomes neutral or
there exists a player whose utility is less than (or more than) $2\ps$.

If there exists $j$ such that $|P_j| \ge 4$,
then we can verify by exhaustive search that
there must exist a player $p_i$, $i \in P_j$, has the utility less than $2\ps$.

If there exists $j$ such that $|P_j| = 3$,
we can assume that the three players choose either $b_{j, 1}, b_{j, 2}$ and
$b_{j,  3}$ or $b_{j, 1}, b_{j, 4}$ and
$b_{j,  5}$; otherwise, a player $p_i$, $i \in P_j$,
must have the utility less than $2\ps$ or one of  $b_{j, 1}, b_{j, 3}$ and $b_{j,4}$
becomes neutral. Without loss of generality, we consider the former case.
Then, the player choosing $b_{j,1}$ dominates $b_{j, 1}$ and $b_{j, 5}$, and
also may dominate $b''$; but, cannot do any of the other vertices
with positive weights
due to the player dominating $b'$ or the one choosing $b_{j, 2}$.
Thus, the player choosing $b_{j,1}$ has the utility at most $s_j + \ps < 2\ps$. 

Suppose there exists $j$ such that $|P_j| = 2$;
let $p$ and $p'$ be the two players.
We consider the following two cases.

\smallskip
\noindent
(i) \emph{Two players choose $b'$ and $b''$.}

In this case, $p$ and $p'$ can dominate
vertices only of $b_{j, 1},b_{j, 2},b_{j, 3},b_{j, 4}$ and $b_{j, 5}$,
the sum of whose weights is $2\ps+s_j < 4\ps$. Thus, one of $p$
and $p'$ has the utility less than $2\ps$.

\smallskip
\noindent
(ii) \emph{Exactly one player chooses $b'$ or $b''$.}

Suppose without loss of generality that a player chooses $b'$,
but no player does $b''$.
Then, the two players must choose either
``$b_{j, 1}$ and $b_{j, 3}$''
or ``$b_{j,4}$ and $b_{j, 5}$.'' However, the player choosing $b_{j, 1}$ (or $b_{j, 5}$)
dominates $b_{j, 1}$ and $b_{j, 5}$, and also may dominate $b''$;
but cannot do any of the other vertices with positive weights
due to the player choosing $b'$. Thus, the utility of the player choosing
$b_{j, 1}$ (or $b_{j, 5}$) is $s_j + \ps < 2\ps$.

\smallskip
\noindent
(iii) \emph{No player chooses $b'$ or $b''$.}

In this case, it suffices to consider the case where
the two players choose either ``$b_{j, 1}$ and $b_{j, 3}$,'' ``$b_{j, 1}$ and
$b_{j, 4}$,'' ``$b_{j, 2}$ and $b_{j, 3}$'' or ``$b_{j, 4}$ and  $b_{j, 5}$,'' since, otherwise, a vertex with a positive integer weight
becomes neutral.

We verify only the case for ``$b_{j, 1}$ and $b_{j, 3}$,''
since the other three cases are similar.
In this case, a player $p_i$ choosing $b_{j, 1}$ dominates
$b_{j,1}$ and $b_{j, 5}$. Furthermore, if $p_i$ does not dominate
$b''$, then the utility is $s_j < 2\ps$.
Otherwise, that is, if $p_i$ dominates $b''$, then $p_i$ faced on one of
the following two situations for any $j' \in [2n] \backslash \{ j\}$:
(iii-1)
if there exists a player choosing one of
$b_{j', 1}, b_{j', 2}, b_{j', 3}, b_{j', 4}$ and $b_{j', 5}$,
then $p_i$ cannot dominate none of them;
and (iii-2)
if no player chooses $b_{j', 1}, b_{j', 2}, b_{j', 3}, b_{j', 4}$ and
$b_{j', 5}$, then either $p_i$ dominates none of them or
$p_i$ does at least $b_{j',4}$.
Therefore, $p_i$ cannot have the utility exactly $2\ps$.

We thus complete the proof of the claim.
\end{proof}

\begin{claim}
$\vec{s}^*$ satisfies (B.1) and (B.2).
\end{claim}
\begin{proof}
Claims~\ref{claim:nonnegative_AtMostTwo} and~\ref{clm:nonnegative_Pj<=1}
immediately imply that
exactly two players must choose $b'$ and $b''$, and hence $\vec{s}^*$ satisfies
(B.2). Therefore, we have $|P_j| = 1$ for every $j \in [2n]$;
otherwise, a vertex with weight $s_j$ such that $|P_j|=0$ becomes neutral.
Consequently, $\vec{s}^*$ satisfies (B.1): otherwise a vertex with a
positive weight becomes neutral, or the game end while
such vertex is not dominated by any player
\end{proof}

\section{Proof of Theorem~\ref{thm:Arbitrary}}
In this section, we prove Theorem~\ref{thm:Arbitrary}: we show that
there exists $S'\subseteq [2n]$
satisfying Eq.~(\ref{eq:partition})
if and only if $(2n+4, G_S, w_S)$ has a Nash equilibrium.
We say that a player $p$ \emph{rules} $B_j$, $j \in [2n]$,
if $p$ dominates all the four vertices
$b_j, b_{j,1}, b_{j,2}$ and $b_{j,3}$ in $B_j$.

\subsection{Main Proof}
We first prove the sufficiency of the theorem.

\begin{lemma}
If there exists $S'\subseteq [2n]$
that satisfies $|S'|=n$  and Eq.~(\ref{eq:partition}), then
$(2n+4, G_S, w_S)$ has a Nash equilibrium.
\end{lemma} 
\begin{proof}
Let $S' \subseteq [2n]$ be a subset satisfying Eq.~(\ref{eq:partition}).
We denote by $j_1, j_2, \dots , j_n \in [2n]$ the indices of the integers
in $S'$, that is, $S=\{ s_{j_1}, s_{j_2}, \dots , s_{j_n}\}$.
Then, the following strategy profile $\vec{s}^*$ is shown to be
a Nash equilibrium.
 
\smallskip
\noindent
{\bf The strategy profile $\vec{s}^*$}: For each $i\in [n]$,
the players $p_{2i-1}$ and $p_{2i}$ choose $b_{j_i, 1}$ and $b_{j_i, 2}$,
respectively. The player $p_{2n+1}$ chooses $b_1'$,
$p_{2n+2}$ does $b_5'$, $p_{2n+3}$ does $a_2$, and
$p_{2n+4}$ does $a_3$, respectively.

\smallskip

We estimate the utility of each player.
Consider the players $p_{2i-1}$ and $p_{2i}$ for every $i \in[n]$.
Since $b_i$ becomes neutral, $p_{2i-1}$ dominates only $b_{i, 1}$, and
$p_{2i}$ does only $b_{i,2}$. Thus,
for every $i\in [n]$,
\begin{eqnarray}\label{eq:arbitrary_U2iI2i-1}
U_{2i-1}(\vec{s}^*)= U_{2i}(\vec{s}^*)=  \ps.
\end{eqnarray}
The player $p_{2n+1}$ dominates only $b_1'$ and $b_2'$. Thus,
\begin{eqnarray}\label{eq:arbitrary_U2n+1}
U_{2n+1}(\vec{s}^*)=\ps+1.
\end{eqnarray}
The player $p_{2n+2}$ dominates $b_4'$ and $b'_5$; besides,
$p_{2n+2}$ rules $B_{j}$ for every $j\in [2n]\backslash S'$. Thus,
 Eq.~(\ref{eq:partition}) implies that
\begin{eqnarray}\label{eq:arbitrary_U2n+2}
U_{2n+1}(\vec{s}^*)=-\ps + \ps+1 + \sum_{j\in [2n]\backslash S'}s_j = \ps+1 .
\end{eqnarray}
The players $p_{2n+3}$ and $p_{2n+4}$ dominate
two vertices in $A$, and hence
\begin{eqnarray}\label{eq:arbitrary_U2n+3_2n+4}
U_{2n+3}(\vec{s}^*)= U_{2n+4}(\vec{s}^*) = \ps+1.
\end{eqnarray}

Below we show that $\vec{s}^*$ is a Nash equilibrium
by verifying Eq.~(\ref{eq:condition_on_Nash}).
We omit the cases where a player changes the strategy to a vertex chosen by another player.

\smallskip
\noindent
(i) $p_1, p_2, \dots , p_{2n}, p_{2n+3}, p_{2n+4}$

Consider an arbitrary player $p_i$, $i \in [2n] \cup \{2n+3, 2n+4\}$.
By Eq.~(\ref{eq:arbitrary_U2iI2i-1}) and
(\ref{eq:arbitrary_U2n+3_2n+4}), it suffices to show that
$U_i(\vec{s}^*_{-i}, v') \le \ps$
for any $v' \in V_A\cup V_B$.
If $v' \in V_A$, then $U_i(\vec{s}^*_{-i}, v') \le \ps$.
If $v' \in \{ b'_2, b'_3\}$, then $U_i(\vec{s}^*_{-i}, v') = -1$.
If $v' = b'_4$, then $p_i$ dominates $b'_3$ and $b'_4$, and also
rules $B_j$ for every $j \in [2n]\backslash S'$. Thus
\[
U_i(\vec{s}^*_{-i}, v') = -1 -\ps + \sum_{j \in [2n]\backslash S'} s_j
\]
and hence Eq.~(\ref{eq:partition}) implies that
\[
U_i(\vec{s}^*_{-i}, v') = -1 -\ps + \ps = -1.
\]
Consider $v' \in B_j$ for some $j \in [2n]$. If $j \in S'$, $p_i$ dominates
either only $b_{j,3}$ or $b_j$ together with $b_{j, 3}$.
Otherwise, since $p_{2n+2}$ chooses $b'_5$,
$p_i$ just rules $B_j$. Thus, $U_i(\vec{s}^*_{-i}, v') \le s_j \le \ps$.

\smallskip
\noindent
(ii) $p_{2n+1}$

By Eq.~(\ref{eq:arbitrary_U2n+1}),
it suffices to show that
$U_i(\vec{s}^*_{-i}, v') \le \ps + 1$ for any $v' \in V_A\cup V_B$.
Similarly to the case (i), 
$U_{2n+1}(\vec{s}^*_{-(2n+1)}, v') = \ps$ if $p_{2n+1}$ chooses
a vertex in $V_A$ or $B_j$ for some $j \in [2n]$.
If $v' \in \{ b'_2, b'_3\}$, then
$p_{2n+1}$ dominates $b'_1, b'_2, b'_3$, and hence $U_{2n+1}(\vec{s}^*_{-(2n+1)}, v') = \ps$.
If $v' = b'_4$, then $p_{2n+1}$ dominates $b'_1, b'_2, b'_3, b'_4$, and also
rules $B_j$ for every $j \in [2n]\backslash S'$. Thus,
by Eq.~(\ref{eq:partition}),
\begin{eqnarray*}
U_{2n+1}(\vec{s}^*_{-(2n+1)}, v') &=& (\ps + 1) -1 -\ps +
\sum_{j \in [2n]\backslash S'} s_j\\
&=& \ps + 1 -1 -\ps + \ps\\
&=& \ps.
\end{eqnarray*}

\smallskip
\noindent
(iii) $p_{2n+2}$

By Eq.~(\ref{eq:arbitrary_U2n+2}),
it suffices to show that
$U_i(\vec{s}^*_{-i}, v') \le \ps + 1$ for any $v' \in V'$.
Similarly to the case (i), if $v' \in V_A$,
$U_{2n+1}(\vec{s}^*_{-(2n+1)}, v') = \ps$.
If $v' \in B_j$ for some $j \in S'$, $p_{2n+2}$ dominates
exactly the same set of the vertices as the one for $\vec{s}^*$, that is,
$p_{2n+2}$ dominates $b'_4, b'_5$, and rules $B_j$ for every
$j \in [2n] \backslash S'$. Thus,
\[
U_{2n+1}(\vec{s}^*_{-(2n+1)}, v') = -\ps + ( \ps+1) + \sum_{j \in S'} s_j = \ps +1,
\]
since $b_j$ becomes neutral for every $j \in S'$.
If $v' \in \{ b'_2, b'_3, b_4\}$, then
$p_{2n+2}$ dominates $b'_3, b'_4, b'_5$, and rules $B_j$ for every
$j \in [2n] \backslash S'$. Thus,
\[
U_{2n+1}(\vec{s}^*_{-(2n+1)}, v') = -1 -\ps + ( \ps+1) + \sum_{j \in S'} s_j = \ps.
\].
\end{proof}

We below prove the necessity of the theorem.

\begin{lemma}
If $(2n+4, G_S, w_S)$ has a Nash equilibrium,
there exists $S'\subseteq [2n]$
that satisfies $|S'|=n$  and Eq.~(\ref{eq:partition}).
\end{lemma} 
\begin{proof}
We show that if $(2n+4, G_S, w_S)$ has a Nash equilibrium, then
the exists $S' \subseteq [2n]$ satisfying Eq.~(\ref{eq:partition}). 
We say that a strategy profile $s$ is \emph{standard} if
$s$ satisfies the following conditions:
\begin{description}
\item[(C.1)] There exists a set $S'\subseteq [2n]$ such that $|S'| = n$ and,
for every $j \in S'$, a player chooses $b_{j, 1}$, a player does $b_{j, 2}$,
and no player choose $b_j$ and $b_{j, 3}$. Moreover, for every
$j\in [2n]\backslash S'$, no player chooses a vertex in $B_j$. 
\item[(C.2)] A player chooses $b'_1$, and a player does $b'_5$. 
\item[(C.3)] A player chooses $a_2$, and a player does $a_3$.
\end{description}
\noindent
We then have the following lemma whose proof will be given in the next section.
\begin{proposition}\label{lem:arbitrary_standard}
If $(2n+4, G_S, w_S)$ has a Nash equilibrium $\vec{s}^*$,
then $\vec{s}^*$ is standard. 
\end{proposition}

Let $\vec{s}^*$ be an arbitrary Nash equilibrium of $(2n+4, G_S, w_S)$;
Proposition~\ref{lem:arbitrary_standard} implies that
$\vec{s}^*$ is standard.
Let $S'=\{ j_1, j_2, \dots , j_n\} \subseteq [2n]$.
We can assume without loss of generality that,
for each $i\in [n]$,
the players $p_{2i-1}$ and $p_{2i}$ choose $b_{j_i, 1}$ and $b_{j_i, 2}$,
respectively; $p_{2n+1}$ chooses $b_1'$,
$p_{2n+2}$ does $b_5'$, $p_{2n+3}$ does $a_2$, and
$p_{2n+4}$ does $a_3$.

Since $p_{2n+1}$ dominates only $b'_1$ and $b'_2$,
\begin{eqnarray}\label{eq:Arbitrary_UB1}
U_{2n+1}(\vec{s}^*) = \ps + 1.
\end{eqnarray}
Besides, if $p_{2n+1}$ changes the strategy to $b'_4$,
$p_{2n+1}$ dominates $b'_1, b'_2, b'_3, b'_4$, and rules $B_j$
for every $j \in [2n]\backslash S'$:
\begin{eqnarray}\label{eq:Arbitrary_UB2}
U_{2n+1}(\vec{s}^*_{-(2n+1)}, b'_4) &=& (\ps + 1) + (-1) +
\ps + \sum_{j \in [2n]\backslash S'} s_j \nonumber \\
&=&  \sum_{j \in [2n]\backslash S'} s_j.
\end{eqnarray}
Since $\vec{s}^*$ is a Nash equilibrium, Eqs.~(\ref{eq:Arbitrary_UB1}) and
(\ref{eq:Arbitrary_UB2}) imply that
\begin{eqnarray}\label{eq:Arbitrary_UB}
\sum_{j \in [2n]\backslash S'} s_j \le \ps + 1.
\end{eqnarray}

On the other hand,
$p_{2n+2}$ dominates $b'_4, b'_5$, and rules 
$B_j$
for every $j \in [2n]\backslash S'$:
\begin{eqnarray}\label{eq:Arbitrary_LB1}
U_{2n+2}(\vec{s}^*) &=& (\ps + 1) + (-\ps) + \sum_{j \in [2n]\backslash S'} s_j \nonumber \\
&=&  1 + \sum_{j \in [2n]\backslash S'} s_j.
\end{eqnarray}
Besides,
if $p_{2n+2}$ changes the strategy to $a_1$,
$p_{2n+2}$ dominates only $a_1$:
\begin{eqnarray}\label{eq:Arbitrary_LB2}
U_{2n+2}(\vec{s}^*_{-(2n+2)}, b'_4) = \ps.
\end{eqnarray}
Therefore, Eqs.~(\ref{eq:Arbitrary_LB1}) and
(\ref{eq:Arbitrary_LB2}) imply that
\begin{eqnarray}\label{eq:Arbitrary_LB}
\ps \le 1 + \sum_{j \in [2n]\backslash S'} s_j.
\end{eqnarray}

By Eqs.~(\ref{eq:Arbitrary_UB}) and (\ref{eq:Arbitrary_LB}),
\[
\ps - 1\le \sum_{j \in [2n]\backslash S'} s_j \le \ps + 1.
\]
Since $s_j$ is even for every $j \in [2n]$, we have
\[
\sum_{j \in [2n]\backslash S'} s_j = \ps,
\]
and hence $S'$ is the desired partition.
\end{proof}

\subsection{Proof of Proposition~\ref{lem:arbitrary_standard}}

Consider an arbitrary Nash equilibrium $\vec{s}^*$ of $(2n+4, G_S, w_S)$.
For every $j\in [2n]$, we define
$\beta_j$ as the number of players
who choose vertices in $B_j$ on $\vec{s}^*$.

Firstly, we prove the following claim.
\begin{claim}\label{claim:arbitrary_C.3}
$\vec{s}^*$ satisfies (C.3).
\end{claim}
\begin{proof}
Suppose $\vec{s}^*$ does not satisfy (C.3).
We derive a contradiction for each of (i), (ii) and (iii) below.

\smallskip
\noindent
(i) \emph{More than three players choose vertices in $A$.}

In this case, there exists a player $p_i$ such that $U_i(\vec{s}^*) \le 1$,
and at most $2n+1$ players choose vertices in $V_B$.

If there exists $j \in [2n]$ such that $\beta _j = 0$, $p_i$ can rule
$B_j$ by changing the strategy to $b_j$. Moreover, if $p_i$ also dominates
$b'_4$, $p_i$ does $b'_5$; and if $p_i$ also dominates
$b_{j'}$ for some $j' \in [2n]$, $p_i$ rules $B_{j'}$. Therefore,
since $s_j$ is even,
\[
U_i(\vec{s}^*_{-i}, b_j) \ge s_j \ge 2.
\]
Otherwise, that is, if $\beta_j \ge 1$ for every $j \in [2n]$, then
there exists $j'\in [2n]$ such that $\beta_{j'}=1$.
Thus, $p_i$ can dominate only either $b_{j', 1}$ or $b_{j', 2}$, and
hence, for some $v' \in \{ b_{j',1}, b_{j',2}\}$,
\[
U_i(\vec{s}^*_{-i}, v') \ge \ps \ge 2.
\]

\smallskip
\noindent
(ii) \emph{exactly two players choose vertices in $A$.}

Since $\vec{s}^*$ does not satisfy (C.3),
we can assume without loss of generality that the two players
choose either ``$a_1$ and $a_2$,'' ``$a_1$ and $a_3$'' or `` $a_1$ and $a_4$.''
For both cases, a player can increase the utility
by choosing $a_2$ or $a_3$.

\smallskip
\noindent
(iii) \emph{at most one player chooses a vertex in $A$.}

In this case, there exists a vertex in $V_A$ such that
any player obtains utility $\ps +1$ by changing the strategy to
the vertex. Thus,
it suffices to show that some player has utility at most $\ps$.
Note that at least $2n+3$ players choose vertices in $V_B$
in this case.

If there exists $j \in [2n]$ such that $\beta_j \ge 2$, then
one of the players has utility at most $\ps$, as desired.
If there exists a pair of $j, j' \in [2n]$ such that $\beta _j = \beta_{j'}=1$,
one of the players (say, the one who choose a vertex in $B_j$) cannot dominate $b'_4$,
and hence its utility is less than or equals to $\ps$, as desired.
Otherwise, that is, if either $\beta_j=0$ for every $j \in [2n]$
or there exists an unique $j\in [2n]$ such that $\beta _j = 1$,
then the other players choose $b'_1, b'_2, b'_3, b'_4, b'_5$,
and hence some player has utility at most zero.
\end{proof}

We below consider the other $2n+2$ players who choose vertices in $V_B$.

\begin{claim}\label{clm:arrbitrary_B}
For every $j\in [2n]$, $\beta _j \le 2$; and there exists at most one $j \in [2n]$
such that $\beta_{j}=1$. 
\end{claim}
\begin{proof}
We first prove that $\beta _j \le 2$ for every $j\in [2n]$.
Note that, by Claim~\ref{claim:arbitrary_C.3},
any player can obtain the utility $\ps$ by changing the strategy to $a_1$. 

Suppose there exists $j$, $j \in [2n]$, such that $\beta_j \ge 4$.
Then, the utility of one of the players is at most $0 < \ps$, as desired.

Suppose there exists $j \in [2n]$ such that $\beta_j = 3$.
If the three players do not choose $b_j, b_{j, 1}$ and $b_{j, 2}$, then 
one of the players has the utility of  at most $0 < \ps$; thus,
let $p_i, p_{i'}$ and  $p_{i''}$ be the players choosing $b_j, b_{j, 1}$ and $b_{j, 2}$, respectively.
If some player other than $p_i, p_{i'}$ and  $p_{i''}$ chooses $b'_4$ or $b'_5$,
the player $p_i$ dominates only $b_j$ and $b_{j,3}$,
and hence $U_i(\vec{s}^*) = -2\ps + s_j < \ps$.
If no player chooses $b'_4$ or $b'_5$, then we have $U_{i'}(\vec{s}^*) = \ps$,
while $U_{i'}(\vec{s}^*_{-i'}, b'_5) = \ps +1$.

We then prove that there exists at most one $j \in [2n]$
such that $\beta_{j}=1$. Suppose there exist $j', j'' \in [2n]$ such that
$\beta_{j'} = \beta_{j''}=1$. Let $p_{i'}$ and $p_{i''}$ be the players
choosing vertices in $B_{j'}$ and $B_{j''}$, respectively.
We consider the following two cases.

\smallskip
\noindent
(i) There exists a player $p_{i}$ choosing $b'_4$.

In this case, we have $U_{i'}(\vec{s}^*) \le \ps$,
while $p_{i'}$ can increase the utility as follows:
If no player chooses $b'_1$, $U_{i'}(\vec{s}^*_{-i'}, b'_1) = \ps +1$;
similarly, if no player chooses $b'_5$, $U_{i'}(\vec{s}^*_{i'}, b'_5) = \ps +1$.
If there are a player choosing $b'_1$ and one choosing
$b'_5$,
\begin{eqnarray*}
U_{i}(\vec{s}^*) &\le& -\ps + \sum_{j \in [2n]\backslash \{ j', j''\}} s_j\\
&\le& -\ps + 2\ps - s_{j'}-s_{j''}\\
&<& \ps.
\end{eqnarray*}

\smallskip
\noindent
(ii) No player chooses $b'_4$.

In this case, one of $p_{i'}$ and $p_{i''}$ does not dominate $b'_4$;
let $p_{i'}$ be such a player. Then, $p_{i'}$ rules $B_{j'}$, but does not
dominate other verticies. Thus, $U_{i'}(\vec{s}^*) = s_{j'} <\ps$.
\end{proof}

Claim~\ref{clm:arrbitrary_B} implies that $\vec{s}^*$ satisfies (C.1), as follows.
\begin{claim}\label{clm:arrbitrary_C.1}
$\vec{s}^*$ satisfies (C.1).
\end{claim}
\begin{proof}
Let ${\cal B} = \{ j \in [2n] \mid \beta_j = 2\}$. We first verify that
$|{\cal B}| = n$, and then show that the two players in $B_j$ for
each $j \in {\cal B}$ choose $b_{j, 1}$ and $b_{j, 2}$.
Recall that, by Claim~\ref{claim:arbitrary_C.3},
any player can obtain the utility $\ps$ by changing the strategy to $a_1$.

If $|{\cal B}| = 0$, Claim~\ref{clm:arrbitrary_B} implies that at least 
$2n+1$ players choose $b'_1, b'_2, b'_3. b'_4, b'_5$.
Thus, some player has the utility at most $0 < \ps$.

If $1 \le |{\cal B}| \le n-1$, Claim~\ref{clm:arrbitrary_B} implies that at least 
three players choose $b'_1, b'_2, b'_3. b'_4, b'_5$, and hence there exists
a player $p_i$ among the three players who dominates neither $b'_1$ nor $b'_5$.
Therefore, since we have $1 \le |{\cal B}|$,
\[
U_i(\vec{s}^*) \le -\ps + \sum_{j \in [2n]\backslash {\cal B}} s_j < \ps.
\]

If $|{\cal B}| = n+1$, then no player chooses $b'_1, b'_2, b'_3. b'_4, b'_5$.
Moreover, since only one player dominates $b'_4$, there exists a player $p_i$
such that $U_i(\vec{s}^*) \le \ps$. However,
$p_i$ can dominates $b'_1, b'_2, b'_3. b'_4, b'_5$ by changing the strategy to
$b'_4$, and thus
\[
(\ps+1)+(-1)+(-\ps)+(\ps+1) = \ps + 1 \le U_i(\vec{s}^*_{-i}, b'_4).
\]

By the argument above, we have $|{\cal B}| = n$, and hence
we have $n$ pairs of players and the other two players.
We below show that
a pair of players
for $B_j$, for an arbitrary index $j \in {\cal B}$,
choose $b_{j, 1}$ and $b_{j, 2}$.
Suppose for the sake of contradiction that
one of the players chooses neither $b_{j, 1}$ nor $b_{j, 2}$.

If a player chooses $b_{j, 3}$, then,
since the other player
chooses $b_j, b_{j,1}$ or $b_{j, 2}$, the
utility of the player is $-\ps < \ps$.

If no player chooses $b_{j, 3}$, there exists a player $p_{i}$ choosing
$b_j$, and $p_{i}$ dominates $b_j$ and $b_{j, 3}$ together with
only one of $b_{j, 1}$ and $b_{j,2}$ in $B_j$.
In addition to the three vertices,
$p_{i}$ does or does not dominate $b'_4$.

In the former case,
there must be no player choosing $b'_3, b'_4$ and $b'_5$.
Consider then the other two players who may choose $b'_1$ and $b'_2$.
If the two players choose $b'_1$ and $b'_2$, one of them has utility zero
which is less than $\alpha$.
If only a player $p_{i'}$ chooses one of $b'_1$ and $b'_2$,
\[
U_{i'}(\vec{s}^*) \le \ps + 1
\]
while
\begin{eqnarray*}
U_{i'}(\vec{s}^*_{-i'}, b'_4) &\ge& (\ps + 1) + (-1) + (-\ps) + 
\sum_{j'\in [2n]\backslash {\cal B}} s_{j'} +\ps + 1\\
&\ge&  \ps + 3.
\end{eqnarray*}

In the latter case, we have
\[
U_{i'}(\vec{s}^*) \le (-\ps + s_j) + (-\ps) + \ps = -\ps +s_j < \ps.
\]

We thus complete the proof.
\end{proof}

By Claims~\ref{claim:arbitrary_C.3} and~\ref{clm:arrbitrary_C.1},
it suffices to consider the other two player choosing vertices in $V_B$.
We complete the proof of Lemma~\ref{lem:arbitrary_standard} with
the following claim.
\begin{claim}
$\vec{s}^*$ satisfies (C.2).
\end{claim}
\begin{proof}
Let $p_i$ and $p_{i'}$ be the two players other than $2n+2$ players
satisfying (C.1) and (C.3). Let $S' \subseteq [2n]$ be the set of
indicies satisfying (C.1). 

We first show that one of $p_i$ and $p_{i'}$ chooses $b'_1$.
Suppose no player chooses $b'_1$; then Claim~\ref{clm:arrbitrary_B}
implies that at least one of $p_i$ and $p_{i'}$ chooses
$b'_2, b'_3, b'_4$ and $b'_5$; assume that $p_i$ is such a player.
Since Claim~\ref{clm:arrbitrary_C.1} implies that any player choosing $B_j$,
$j \in S'$, has utility $\ps$; but, the player can 
obtain, by changing the strategy to $b'_1$, $\ps +1$
and does not dominate $b'_3$ due to $p_i$.

Let $p_i$ be the player choosing $b'_1$. We then show that $p_{i'}$
chooses $b'_5$. If $p_{i'}$ chooses $b'_2, b'_3$ or $b'_4$,
then $p_{i'}$ can increase the utility by changing the strategy to $b'_5$
by avoiding $b'_3$. Otherwise, that is, if $p_{i'}$ chooses a vertex
in $B_j$, $j \in [2n] \backslash S'$, then $p_i$ cannot dominate
$b'_4$ and $b'_5$, and hence
\[
U_i(\vec{s}^*_{-i}) \le \ps +1.
\]
However, by changing the strategy to $b'_4$, $p_i$ dominates
$b'_1, b'_2, b'_3, b'_4$ and $b'_5$ together with $B_{j'}$
for every $j'\in [2n]\backslash (S'\cup \{j\})$.
Therefore, since $|S'\cup \{ j\}| < 2n$,
\[
U_i(\vec{s}^*_{-i}, b'_4) \ge 
(\ps+1)+(-1)+(-\ps)+ \sum_{j'\in [2n]\backslash (S'\cup \{j\})}
s_{j'} +(\ps+1) > \ps + 1.
\]
\end{proof}

\section{Proof of Lemma~\ref{lem:decompose_players}}

Since the necessity is obvious, we only verify the sufficiency.
That is, we show that if $(k, G, w)$ has a Nash equilibrium, then
there exist $\kappa_1, \kappa_2, \dots , \kappa_m$ and $t$
such that $k = \sum_{j=1}^m \kappa_j$ and
$P_j$ admits $\kappa_j$ players with a common boundary $t$ for every $j\in [m]$.
Note that, if a strategy profile  $\vec{s} = (\vec{s}_1, \vec{s}_2,
\dots , \vec{s}_m)$ is a Nash equilibrium, then 
it trivially holds that 
$\suki_{P_j}(\vec{s}_j) \le \umin_{P_j}(\vec{s}_j)$
for any $j \in [m]$.

Consider an arbitrary Nash equilibrium $\vec{s}$, and let
$\vec{s}_1, \vec{s}_2, \dots , \vec{s}_m$ be strategy
profiles for $P_1, P_2, \dots , P_m$, and
$\kappa_1, \kappa_2, \dots , \kappa_m$ be the numbers of players choosing
vertices in $P_1, P_2, \dots , P_m$, respectively.
Suppose for sake of contradiction that no
common boundary $t$ exists for $P_1, P_2, \dots , P_m$ with $\kappa_1, \kappa_2, \dots , \kappa_m$.
Let $t$ be an arbitrary integer such that there exists $j$ satisfying
$\suki_{P_j}(\vec{s}_j) \le t \le \umin_{P_j}(\vec{s}_j)$.
Then there is at least an index $j'$ such that either
$\umin_{P_{j'}} (\vec{s}_{j'}) < t$ or $t < \suki_{P_{j'}}(\vec{s}_{j'})$.
Let $J'_{<t}$ be the set of indices satisfying
the former inequality $\umin_{P_{j'}} (\vec{s}_{j'}) < t$,
and $J'_{t<}$ be the set of indices satisfying the latter inequality
$\umin_{P_{j'}} (\vec{s}_{j'}) < t$.

Consider the case for  $J'_{<t} \neq \emptyset$ and $J'_{t<} \neq \emptyset$.
In this case, there exists a pair of indices $j'_1 \in J_{<t}$
and $j'_2 \in J'_{t<}$ such that
$\umin_{P_{j'_1}} (\vec{s}_{j'_1}) < \suki_{P_{j'_2}}(\vec{s}_{j'_2})$.
Thus, a player in $P_{j'_1}$ can increase the utility by changing
the strategy to a vertex in $P_{j'_2}$.

Consider the other case $J'_{<t} = [m]$ or $J'_{<t} = [m]$.
If $J'_{<t} = [m]$, there exists an index $j' \in [m]$ such that
$\umin_{P_{j'}} (\vec{s}_{j'}) < \suki_{P_{j}}(\vec{s}_{j})$;
otherwise we can choose
\[
t = \min _{j' \in [m]} \umin_{P_{j'}} (\vec{s}_{j'})
\]
as the common boundary.
Thus, a player in $P_{j'}$ can increase the utility by changing
the strategy to a vertex in $P_{j}$. Similarly,
if $J'_{t<} = [m]$, there exists an index $j' \in [m]$ such that
$\umin_{P_{j}} (\vec{s}_{j}) < \suki_{P_{j'}}(\vec{s}_{j'})$;
otherwise we can choose
\[
t = \max _{j' \in [m]} \suki_{P_{j'}} (\vec{s}_{j'})
\]
as the common boundary.
Thus, a player in $P_{j}$ can increase the utility by changing
the strategy to a vertex in $P_{j'}$.

\section{Proof of Lemma~\ref{lem:ForestsOfPaths_weighted}}

We give such an algorithm as a proof of Lemma~\ref{lem:ForestsOfPaths_weighted}.
Suppose that we are given a path $P$ of $n$ vertices, together with its weight function $w$ and a nonnegative integer $t$.

We first note that it is easy to determine whether $P$ admits zero player with the boundary $t$. 
In this case, there is only one strategy profile $\vec{s} = \emptyset$ for $(0, P, w)$, and $\umin_{P} (\vec{s}) = + \infty \ge t$.
Thus, $0 \in K$ if and only if $\suki_{P}(\vec{s})  = \sum_{v \in V(P)} w(v) \le t$. 
Clearly, this can be done in $O(n)$ time. 
In the following, we thus assume that at least one player exists in the game $(\kappa, P, w)$. 
\smallskip

\noindent
{\bf Idea and definitions.}

We label the vertices $v_1, v_2, \ldots, v_n$ in the path $P$ from left to right.
Our algorithm employs a dynamic programming approach based on the number $\kappa$ of players in $P$, and mainly keeps truck of the first and second rightmost vertices that are chosen by some of the $\kappa$ players in $P$ together with the numbers of players who chose them. 
However, it should be noted that, because we will add a new $(\kappa+1)$-th player to the right side of the currently rightmost player, the utility of the rightmost player may change according to the addition. 
Therefore, as a ``partial'' solution, we introduce the notion of ``Nash sub-equilibrium'' for a weighted path, as follows.

Let $\vec{s}$ be a strategy profile for the game $(\kappa, P, w)$ such that $\kappa \ge 1$.
We denote by $r_{\vec{s}}$ the index $1 \le r_{\vec{s}} \le n$ such that $v_{r_{\vec{s}}}$ is the rightmost vertex chosen by at least one player in $\vec{s}$.
Let $\akifixvset_{\vec{s}}$ be the set of all players in $\vec{s}$ who did not choose $v_{r_{\vec{s}}}$; 
we denote by $\umin(\vec{s}; \akifixvset_{\vec{s}})$ the minimum utility over all players in $\akifixvset_{\vec{s}}$.
For two integers $x$ and $y$, $0 \le y < x \le n$, let $\suki(\vec{s}; y,x) = \max_{y \le j \le x} U_{\kappa+1}(\vec{s}+v_j)$. 
In other words, $\suki(\vec{s}; y,x)$ denotes the potential of the maximum utility under $\vec{s}$ that can be expected to gain by an extra player other than the $\kappa$ players when the extra player chooses a vertex from $v_y, v_{y+1}, \ldots , v_{x}$.
Then, a strategy profile $\vec{s}$ for the game $(\kappa, P, w)$ is called a {\em Nash sub-equilibrium with a boundary $t$} if and only if 
\begin{description}
\item[(1)] $\suki(\vec{s}; 1, r_{\vec{s}}) \le t$;
\item[(2)] $\umin(\vec{s}; \akifixvset_{\vec{s}}) \ge t$; and
\item[(3)] $U_j(\vec{s}_{-j}, v') \le U_j(\vec{s})$ for each player $j \in \akifixvset_{\vec{s}}$ and $v' \in \{v_1, v_2, \ldots , v_{r_{\vec{s}}}\}$.
\end{description}

We classify the Nash sub-equilibriums with respect to the first and second rightmost vertices chosen by players. 
For notational convenience, we introduce a dummy vertex $v_0$. 
Let $\kappa$ be an integer such that $1 \le \kappa \le 2n$.
Then, we say that $\bigl((x, a), (y, b)\bigr)$ is a {\em feasible set on $\kappa$} if 
\begin{description}
\item[$\bullet$] $x$ and $y$ are indices of vertices in $V(P) \cup \{v_0 \}$ such that $0 \le y < x \le n$; and 
\item[$\bullet$] $a$ and $b$ are integers such that $1 \le a \le \kappa$ and $0 \le b \le \kappa-a$.
\end{description}
For a feasible set $\bigl((x, a), (y, b)\bigr)$ on $\kappa$, we say that the path $P$ admits a {\em Nash $\bigl(\kappa, (x, a), (y, b)\bigr)$-sub-equilibrium with the boundary $t$} if and only if 
there exist a strategy profile $\vec{s}$ for $(\kappa, P, w)$ such that 
\begin{description}
\item[(1)] $\vec{s}$ is a Nash sub-equilibrium with the boundary $t$;
\item[(2)] exactly $a$ players in $\vec{s}$ choose $v_x$;
\item[(3)] if $\kappa - a > 0$, then exactly $b$ players in $\vec{s}$ choose $v_y \in \{ v_1, v_2, \ldots, v_{x-1} \}$; 
\item[(4)] if $\kappa -a = 0$, then $(y, b) = (0,0)$; and
\item[(5)] if $\kappa-a-b > 0$, then $v_y \in \{ v_2, v_3, \ldots, v_{x-1} \}$ and each of the $\kappa-a-b$ players in $\vec{s}$ chooses a vertex from $v_1, v_2, \ldots, v_{y-1}$.
\end{description}
Then, $\bigl((x, a), (y, b)\bigr)$ indicates the first and second rightmost vertices $v_x$ and $v_y$ chosen by players in $\vec{s}$ and the numbers of players $a$ and $b$ who chose $v_x$ and $v_y$, respectively. 
Therefore, the utility $U_j(\vec{s})$ of a player $j$ who chose $v_x$ can be computed only by $\bigl((x, a), (y, b)\bigr)$;
let $U_j\bigl((x, a), (y, b)\bigr) = U_j(\vec{s})$ for a Nash $\bigl(\kappa, (x, a), (y, b)\bigr)$-sub-equilibrium. 
Furthermore, $U_j(\vec{s}_{-j}, v')$ can be computed only by $\bigl((x, a), (y, b)\bigr)$ if  $j$ is a player who chose $v_x$ and $v' \in \{v_{y+1}, v_{y+2}, \ldots, v_n\}$;
let $U_j\bigl((x, a-1), (y, b); v'\bigr) = U_j(\vec{s}_{-j}, v')$. 
Similarly, $\suki(\vec{s}; x,n)$ can be computed by $\bigl((x, a), (y, b)\bigr)$, indeed only by $(x, a)$;
let $\suki'((x,a), n) = \suki(\vec{s}; x,n)$ for a Nash $\bigl(\kappa, (x, a), (y, b)\bigr)$-sub-equilibrium.
Then, we have the following lemma.
\begin{lemma} \label{lem:subNash}
Let $\vec{s}$ be a strategy profile $\vec{s}$ for the game $(\kappa, P, w)$ which is a Nash $\bigl(\kappa, (x, a), (y, b)\bigr)$-sub-equilibrium with the boundary $t$ for some feasible set $\bigl((x, a), (y, b)\bigr)$ on $\kappa$.
Then, $\vec{s}$ is a Nash equilibrium for $P$ with the boundary $t$ if and only if 
\begin{description}
\item[(1)] $\suki'((x,a), n) \le t${\rm ;}
\item[(2)] $U_j\bigl((x, a), (y, b)\bigr) \ge t$ for a player $j$ who chose $v_x${\rm ;} and
\item[(3)] $U_j\bigl((x, a-1), (y, b); v'\bigr) \le U_j\bigl((x, a), (y, b)\bigr)$ for a player $j$ who chose $v_x$ and $v' \in \{v_{y+1}, v_{y+2}, \ldots, v_n\}$.
\end{description}
\end{lemma}
\begin{proof}
The only-if-part clearly holds, and hence we prove the if-part. 
Since $\vec{s}$ is a Nash $\bigl(\kappa, (x, a), (y, b)\bigr)$-sub-equilibrium with a boundary $t$ for the game $(\kappa, P, w)$, it satisfies
\begin{description}
\item[(1')] $\suki(\vec{s}; 1, x) \le t$;
\item[(2')] $\umin(\vec{s}; \akifixvset_{\vec{s}}) \ge t$; and
\item[(3')] $U_j(\vec{s}_{-j}, v'') \le U_j(\vec{s})$ for each player $j \in \akifixvset_{\vec{s}}$ and $v'' \in \{ v_1, v_2, \ldots , v_{x} \}$.
\end{description}
Then, Conditions~(1') and (1) ensure that $\suki_{P}(\vec{s}) \le t$ holds.
Similarly, Conditions~(2') and (2) ensure that $\umin_{P}(\vec{s}) \ge t$ holds.
Furthermore, Conditions~(2') and (1) imply that no player $j$ in $\akifixvset_{\vec{s}}$ increases the utility $U_j(\vec{s})$ by moving to a vertex in $v_{x}, v_{x+1}, \ldots, v_n$. 
Similarly, by Conditions~(1') and (2) no player $j$ who chose $v_x$ moves to a vertex in $v_1, v_2, \ldots, v_{x-1}$. 
Together with Conditions (3') and (3), we thus conclude that $\vec{s}$ is a Nash sub-equilibrium with a boundary $t$. 
\end{proof}

We now define our DP table. 
For a feasible set $\bigl((x, a), (y, b)\bigr)$ on an integer $\kappa \in [2n]$, let 
\[
	f_t\bigl(\kappa, (x, a), (y, b)\bigr) = \left\{
			\begin{array}{ll}
			\TRUE & ~~~\mbox{if $P$ admits a Nash $\bigl(\kappa, (x, a), (y, b)\bigr)$-sub-}\\
					  & ~~~\mbox{equilibrium with the boundary $t$}; \\
			\FALSE & ~~~\mbox{otherwise}.
			\end{array} \right.
\]


\noindent
{\bf Algorithm.}

As the initialization, we first compute $f_t \bigl(\kappa, (x, \kappa), (0, 0) \bigr)$ for each integer $\kappa \in [2n]$ and $x \in [n]$.
Let $\vec{s}$ be the strategy profile for $(\kappa, P, w)$ such that all the $\kappa$ players choose $v_x$.
Then, by the brute-force manner, we can check in $O(n^2)$ time whether $\vec{s}$ is a Nash $\bigl(\kappa, (x, \kappa), (0, 0) \bigr)$-sub-equilibrium with the boundary $t$ or not. 
Therefore, $f_t \bigl(\kappa, (x, \kappa), (0, 0) \bigr)$ can be computed in $O(n^4)$ time for all integers $\kappa \in [2n]$ and $x \in [n]$.
\medskip

We then compute $f_t\bigl(\kappa, (x, a), (y, b)\bigr)$ for each integer $\kappa \in [2n]$ and its feasible sets $\bigl((x, a), (y, b)\bigr)$ such that $b \ge 1$.
Suppose that $P$ admits a Nash $\bigl(\kappa, (x, a), (y, b)\bigr)$-sub-equilibrium $\vec{s}$ with the boundary $t$. 
Then, by removing the $a$ players who chose $v_x$, we can obtain another Nash $\bigl(\kappa-a, (y, b), (z, c)\bigr)$-sub-equilibrium $\vec{s}'$ with the boundary $t$ for some feasible set $\bigl((y, b), (z, c)\bigr)$ on $\kappa-a$;
recall that the utility is not taken into account for the $b$ players who chose $v_y$ in $\vec{s}'$.
Therefore, $f_t\bigl(\kappa, (x, a), (y, b)\bigr) = \TRUE$ if and only if there exists at least one feasible set $\bigl((y, b), (z, c)\bigr)$ on $\kappa-a$ such that 
\begin{description}
\item[(1)] $f_t \bigl(\kappa-a, (y, b), (z, c) \bigr) = \TRUE$;
\item[(2)] the utility of each player $j$ who chose $v_y$ is at least $t$ under the condition that $a$ players chose $v_x$, $b$ players chose $v_y$ and $c$ players chose $v_z$;
\item[(3)] each player $j$ who chose $v_y$ does not move to a vertex in $v_{z}, v_{z+1}, \ldots, v_{x}$ to increase its utility under the condition that $a$ players chose $v_x$, $b$ players chose $v_y$ and $c$ players chose $v_z$; and 
\item[(4)] the utility of an extra player $j'$ is at most $t$ under the condition that $a$ players chose $v_x$, $b$ players chose $v_y$, and $j'$ chose a vertex from $v_{y}, v_{y+1}, \ldots, v_{x}$.
\end{description}
For each feasible set $\bigl((x, a), (y, b)\bigr)$ on $\kappa$, we check all feasible sets $\bigl((y, b), (z, c)\bigr)$ on $\kappa-a$.
The conditions (1)--(4) above can be checked in $O(n^2)$ time for one feasible set $\bigl((y, b), (z, c)\bigr)$. 
Then, we can compute $f_t\bigl(\kappa, (x, a), (y, b)\bigr)$ in time $O(n^9)$ for all feasible sets $\bigl((x, a), (y, b)\bigr)$ on $\kappa$.  



%


\medskip

Finally, we compute the desired set $K$, as follows:
An integer $\kappa$ is in $K$ if and only if there exists a feasible set $\bigl((x, a), (y, b)\bigr)$ on $\kappa$ such that $f_t\bigl(\kappa, (x, a), (y, b)\bigr) = \TRUE$ and Conditions~(1)--(3) of Lemma~\ref{lem:subNash} are satisfied. 
This can be done in time $O(n^7)$. 

In this way, our algorithm runs in time $O(n^9)$ in total. 
This completes the proof of Lemma~\ref{lem:ForestsOfPaths_weighted}.

\section{Proof of Lemma~\ref{lem:ForestsOfPaths}}
Let $P$ be a path of $n$ vertices.
We denote by $\kappa$ and $t$ by the number of players choosing vertices in
$P$ and a boundary $t$, respectively.

\smallskip
\noindent
\emph{Proof for (1)}: 
If $P$ admits $\kappa = 0$ with $t$, then $n \le t$ holds; otherwise, we have $t+1 \le \suki_P(\vec{s})$.
If $n \le t$ holds, then $P$ clearly admits $\kappa = 0$ with $t$.

\smallskip
\noindent
\emph{Proof for (2)}: If $P$ admits $\kappa = 1$ with $t$, then
$t \le n \le 2t+1$ holds, as follows.
If $n \le t-1$, we have $\umin(\vec{s})_P \le t-1$; and if $2t+2 \le n$, then
$t+1 \le \suki_P(\vec{s})$ no matter what vertex the player choose.

If $t \le n \le 2t+1$ holds, then $P$ admits $\kappa = 1$ with $t$
by taking a strategy profile in which the player
choose the center for odd $n$ (or one of the two centers for even $n$)
of $P$.

\smallskip
\noindent
\emph{Proof for (3)}: If $P$ admits $\kappa = 2$ with $t$, $2t \le n \le 2t+2$ holds, as follows.
If $n \le 2t-1$, then one of the two player has utility less than $t$:
$\umin(\vec{s})_P \le t-1$. If $2t+3 \le n$, then
$t+1 \le \suki_P(\vec{s})$, since the two players must
choose two neighboring vertices in any Nash equilibrium $\vec{s}$.

If $2t \le n \le 2t+2$ holds,  $P$ admits $\kappa = 2$ with $t$, as follows.
Consider a strategy profile $\vec{s}$ in which the two players
choose the center and a neighbor of $P$ if  $n$ is odd $n$,
or do the two centers if $n$ is even.
Then, clearly $\vec{s}$ is a Nash equilibrium, and
$\suki_P(\vec{s}) \le t$ and $t \le \umin_P(\vec{s})$.

\smallskip
\noindent
\emph{Proof for (4)}:
The proof of a similar statement is given in~\cite{Takehara2012_2}, but we give it
for completeness.
 
If $P$ admits $\kappa = 3$ with $t$, then $t=1$ and $n = 3, 4$, or $5$, as follows.
If $n \le 2$, then at least one of the three players has utility zero:
$\umin_P(\vec{s})$ is not positive.
If $n = 3, 4$, or $5$, we can easily verify that
every Nash equilibrium $\vec{s}$ satisfies $\suki_P(\vec{s}) \in \{0, 1\}$ and $\umin_P(\vec{s}) = 1$,
and hence $t=1$.
If $6 \le n$, we show that no Nash equilibrium
exists; there are two cases to consider.
If the three vertices are neighboring each other,
then the player in the middle
must have utility one, but can increase their utility by changing the strategy
to the one just to the right (of left) of a vertex outside.
If the three vertices are not neighboring,
a player choosing leftmost or rightmost vertex
can increase the utility by changing the strategy to the one
next to the middle.

If $t=1$ and $n = 3, 4$, or $5$, then
$P$ admits $\kappa = 3$ with $t$ by taking a strategy profile in which
the three players choose the three centers of $P$.

\smallskip
\noindent
\emph{Proof for (5)}:
We consider only the case for even $\kappa$($\ge 4$),
since the proof for the other case is similar.

We define some terms for the proof.
Let $\kappa$ be the number of players,
$P$ be the path of $n$ vertices, and
$w$ be the weight function.
Let $s^{(1)}, s^{(2)}, \dots , s^{(\kappa)}$ be the vertices
the players $p_1, p_2, \dots , p_\kappa$ choose.
Without loss of generality, for every $i$, $ 1 \le i \le \kappa$,
we can assume that $s^{(i)}$ is located in the $i$th position:
$s^{(1)}, s^{(2)}, \dots , s^{(i-1)}$ are to the left of $s^{(i)}$.
We denote by $\delta_0$ the number of vertices to the left of $s^{(1)}$;
for every $i$, $0 \le i \le \kappa - 1$,
by $\delta_i$ the number of vertices between $s^{(i)}$ and $s^{(i+1)}$;
and by $\delta_\kappa$ the number of vertices to the right of $s^{(\kappa)}$.
Clearly, $\delta_0, \delta_1, \dots , \delta_\kappa$ determine the number
\begin{eqnarray}\label{eq:(5)_n=delta_i+k}
n = \sum_{i=0}^\kappa \delta_i + \kappa
\end{eqnarray}
and the strategy profile $\vec{s}$ for $(\kappa, P, w)$.
Note that, if $\vec{s}$ with $\delta_0, \delta_1, \dots , \delta_\kappa$
is a Nash equilibrium, then $\delta_1 = \delta_{\kappa-1} = 0$;
otherwise, the first (or the last player) increases their utility
by changing the strategy to the one next to $s^{(2)}$($s^{(\kappa-1)})$.
Moreover, we can easily observe that
if  $\umin(\vec{s}) \le t \le \suki_P(\vec{s})$ holds, we have
$t-1 \le \delta_0 \le t$ and $t-1 \le \delta_\kappa \le t$.

\bigskip
\noindent
[$\Rightarrow$]

Consider first even $\kappa$.
We here prove that, if the number $n, \kappa$ and $t$ do not satisfy
Eq.~(\ref{eq:FofP_k >= 4}), then $(\kappa, P, w)$ has no Nash equilibrium
such that $\umin_P(\vec{s}) \le t \le \suki_P(\vec{s})$.
If $n < \kappa t$, then we have $n/\kappa < t$, and hence
no desired strategy profile with a boundary $t \le \umin_P(\vec{s})$ exists.

Suppose $(2\kappa-4)t+\kappa < n$, and
the desired strategy profile $\vec{s}$
with $\delta_0, \delta_1, \dots , \delta_\kappa$ exists.
Since $\vec{s}$ is a Nash equilibrium, we have $\delta_1 = \delta_{\kappa-1} = 0$,
$\delta_0 \le t$ and $\delta_\kappa \le t$. Thus, it holds that
\[
(2\kappa-4)t + \kappa - 2t =2t(\kappa-3) + \kappa
<\sum_{i=2}^{\kappa-2} \delta_i + \kappa.
\]
Therefore, there exists an index $i$, $2 \le i \le \kappa-2$, such that
$2t < \delta_i$. Thus, an extra player can gain the utility more than $t$
by choosing the vertex to the right of $s^{(i)}$, and hence
$t < \suki(\vec{s})$, which is a contradiction.

Consider then odd $\kappa$.
Suppose for sake of contradiction that $n < (\kappa +1)t-1$,
and
the desired strategy profile $\vec{s}$
with $\delta_0, \delta_1, \dots , \delta_\kappa$ exists.
By Eq.~(\ref{eq:(5)_n=delta_i+k}), we have
\[
\sum_{i=0}^\kappa \delta_i + \kappa < (\kappa +1)t-1,
\]
and hence
\begin{eqnarray}\label{eq:(5)_delta_1-k}
\sum_{i=0}^\kappa \delta_i < (\kappa+1)(t-1).
\end{eqnarray}
On the other hand,
since $\vec{s}$ is a Nash equilibrium with a boundary $t$, we have
\begin{eqnarray}\label{eq:(5)_012}
t-1 \le \delta_0, \quad \mbox{and} \quad \delta_1 = 0;
\end{eqnarray}
for any even $i$ satisfying  $2 \le i \le \kappa - 1$, we have
\begin{eqnarray}\label{eq:(5)_3_k-3}
2t-2 \le \delta_{i-1} + \delta_i;
\end{eqnarray}
and
\begin{eqnarray}\label{eq:(5)_k-2k-1k}
\delta_{\kappa-1} = 0, \quad \mbox{and} \quad t-1 \le \delta_{\kappa}.
\end{eqnarray}
Since $\kappa$ is odd, we can combine
the inequalities in~(\ref{eq:(5)_012})$-$(\ref{eq:(5)_k-2k-1k}), and obtain
\begin{eqnarray*}\label{eq:(5)_delta_2-k-2}
(t-1)  + (2t-2)\left(\frac{\kappa-3}{2} + 1\right)  + ( t-1)
= (\kappa +1)(t-1) \le \sum_{i=0}^{\kappa} \delta_i,
\end{eqnarray*}
which contradicts Eq.~(\ref{eq:(5)_delta_1-k}).

The proof for the inequality $n \le (2\kappa -4) + \kappa$ is same as
the one for the case where $n$ is even.


\bigskip
\noindent
[$\Leftarrow$]

We first consider even $\kappa$.
We prove by construction that, if the number $n, \kappa$ and $t$ satisfy Eq.~(\ref{eq:FofP_k >= 4}),
then $P$ admits $\kappa$ players with the boundary $t$,
that is, there exists
a strategy profile $\vec{s}$ such that
$\vec{s}$ is a Nash equilibrium for $(\kappa, P, w)$  and
$\suki_P(\vec{s}) \le t \le \umin_P(\vec{s})$.

For a given $n$ satisfying Eq.(\ref{eq:FofP_k >= 4}),
we can construct the desired strategy profile $\vec{s}$
with the following
designated $\delta_0, \delta_1, \dots , \delta_\kappa$.

\smallskip
\noindent
{\bf The strategy profile $\vec{s}$}: 
For every $i$, $0 \le i \le \kappa$, the parameter $\delta_i$ satisfies
\begin{eqnarray}\label{eq:(5)_definition_of_s_even}
\delta_i \in \left\{
\begin{array}{ll}
\{ t-1, t\} & \mbox{ if } i=0;\\
\{ 0 \} & \mbox{ if } i=1;\\
\{ 2t-2, 2t-1, 2t\} & \mbox{ if $2 \le i \le \kappa-2$ and $i$ is odd};\\
\{ 0, 1, \dots , 2t\} & \mbox{ if $2 \le i \le \kappa-2$ and $i$ is even};\\
\{ 0 \} & \mbox{ if } i=\kappa-1;\\
\{ t-1, t\} & \mbox{ if } i=\kappa.
\end{array}
\right.
\end{eqnarray}
Besides, there is an index $i^*$, $0 \le i^* \le \kappa$ such that
$\delta_0, \delta_1, \dots \delta_{i^*-1}$ have their maximum values, while
$\delta_{i^*+1}, \delta_{i^*+2}, \dots \delta_{\kappa}$ have their minimum values. 

\smallskip

If $i^*=0$ and $\delta_0$ has the minimum value
(i.e., $\delta_{i^*} = \delta_0 = t-1$), then the strategy determines
\begin{eqnarray}\label{eq:FofP_kt}
n &=& (t-1) + \left( \frac{\kappa-4}{2} + 1 \right)(2t-2) + (t-1) + \kappa \nonumber\\
&=& \kappa t.
\end{eqnarray}
If $i^*=\kappa$ and $\delta_\kappa$ has the maximum value
(i.e., $\delta_{i^*} = \delta_\kappa = t$), then the strategy determines
\begin{eqnarray}\label{eq:FofP_(2k-4)t+k}
n &=& t + \left( \kappa-4 + 1 \right)2t + t + \kappa \nonumber\\
&=& (2\kappa -4)t+\kappa.
\end{eqnarray}
To obtain $P$ of $n$ vertices,
we can increment the number of vertices one
by one from Eq.~(\ref{eq:FofP_kt}) to Eq.~(\ref{eq:FofP_(2k-4)t+k})
according to the following procedure: if $\delta_{i^*}$
is less than its maximum value, 
we just increment $\delta_{i^*}$ by one; otherwise,
we increment $\delta_{i^*+1}$ by one, and
$i^*+1$ is now considered to be new $i^*$. It now remains
to show that $\vec{s}^*$ is always a Nash equilibrium and
$\suki_P(\vec{s}) \le t \le \umin_P(\vec{s})$.

We can prove that $\vec{s}^*$ is always a Nash equilibrium, as follows.
For every $i$, $ 1 \le i \le \kappa$, we can easily observe that
\begin{eqnarray}\label{eq:FofP_U_i}
t \le U_i(\vec{s});
\end{eqnarray}
and for every $i$, $0 \le i \le \kappa$, we have $\delta_i \le 2t$.
Thus, any player $p_i$ cannot obtain the utility more than $t$ 
by changing the strategy to one of $\delta_{i'}$ vertices
between $s^{(i')}$ and $s^{(i'+1)}$
where
$i'  \le i-2$ or $i+1 \le i'$. Complementally,
if $p_i$ changes the strategy to one of $\delta_{i-1} + \delta_{i} + 1$
vertices between $s^{(i-1)}$ and $s^{(i+1)}$,
the utility of $p_i$ is always $\lfloor (\delta_{i-1}+\delta_i)/2 \rfloor$,
since our construction implies that one of $\delta_{i-1}$ and $\delta_i$ is even. Thus, $\vec{s}$ is a Nash equilibrium.

Similarly,
for every $i$, $0 \le i \le \kappa$, we have $\delta_i \le 2t$, and hence
$\suki_P(\vec{s}) \le t$, while Eq.~(\ref{eq:FofP_U_i}) implies that
$t \le \umin_P(\vec{s})$.

For the case where $\kappa$ is odd,
we can verify, in the same way, the statement by replacing (\ref{eq:(5)_definition_of_s_even})
with
\begin{eqnarray}
\delta_i \in \left\{
\begin{array}{ll}
\{ t-1, t\} & \mbox{ if } i=0;\\
\{ 0 \} & \mbox{ if } i=1;\\
\{ 2t-2, 2t-1, 2t\} & \mbox{ if $2 \le i \le \kappa-3$ and $i$ is odd};\\
\{ 0, 1, \dots , 2t\} & \mbox{ if $2 \le i \le \kappa-3$ and $i$ is even};\\
\{ 2t-2, 2t-2, 2t\} & \mbox{ if } i=\kappa-2;\\
\{ 0 \} & \mbox{ if } i=\kappa-1;\\
\{ t-1, t\} & \mbox{ if } i=\kappa.
\end{array}
\right.
\end{eqnarray}






\end{document}